\documentclass[a4paper,twocolumn,11pt,allowtoday,unpublished]{quantumarticle}
 \pdfoutput=1
\usepackage{graphicx,color}
\usepackage{amsmath,amssymb,amsfonts,amsthm}
\usepackage{multirow}
\usepackage{multicol}
\usepackage{bm}
\usepackage[pdftex,colorlinks=true, linkcolor = blue, citecolor=blue,urlcolor=blue, bookmarksnumbered=true, bookmarksopen=true]{hyperref}
\usepackage[normalem]{ulem}
\usepackage[numbers]{natbib}
\usepackage{braket}
\usepackage{mathtools}
\DeclarePairedDelimiter\ceil{\lceil}{\rceil}

\makeatletter
\def\BState{\State\hskip-\ALG@thistlm}
\makeatother

\newtheorem{theorem}{Theorem}

\newtheorem*{remark*}{Remark}

\begin{document}

\title{Nearly optimal polynomial approximations for the quantum singular value transform}

\author{E. Rule}
\orcid{0000-0003-1316-0970}
\affiliation{
  Theoretical Division, Los Alamos National Laboratory, Los Alamos, New Mexico 87545, USA}
  
\date{\today}

\begin{abstract}
We introduce polynomial approximations of the even and odd step functions on the interval $[-1,1]$ with simple Chebyshev coefficients, making their numerical implementation straightforward. We derive rigorous error bounds and demonstrate that these polynomials are \textit{nearly optimal} in the sense that their error deviates from the theoretically optimal error by a multiplicative factor that grows logarithmically with the polynomial order. From these polynomials, we derive related nearly optimal polynomial approximations that can be used to perform quantum phase estimation, linear amplitude amplification, eigenvalue thresholding, and other quantum algorithms using the quantum singular value transform.
\end{abstract}
\maketitle

\section{Introduction}
Quantum signal processing \cite{Low:2016jxt,PhysRevLett.118.010501} (QSP) and the quantum singular value transform \cite{Gilyen:2018khw} (QSVT) are powerful computational frameworks that unify many of the most celebrated quantum algorithms \cite{PRXQuantum.2.040203}, including Hamiltonian simulation, amplitude amplification, and quantum phase estimation (QPE). Beginning from a block-encoding of an arbitrary matrix $A$, QSVT allows one to apply a polynomial function $A\rightarrow p(A)$, where the notion of $p(A)$ is understood in terms of the singular values of $A$ as detailed below. Such technology enables one to approximately encode any transformation $f(A)$ that itself can be suitably approximated by a polynomial to within some target error. For example, applying QSVT with a polynomial that approximates the function $1/x$ allows one to perform matrix inversion \cite{doi:10.1137/16M1087072,Sunderhauf:2025lbc}.

Any matrix $A$ has a singular value decomposition $A=W \Sigma V^\dagger$, where $W$ and $V$ are unitary matrices and $\Sigma$ is a diagonal matrix of non-negative real numbers $\{\sigma_k\}$ known as the singular values of $A$. In the standard implementation of QSVT \cite{PRXQuantum.2.040203}, one begins from a block-encoding of a matrix $A$ into a unitary $U$ as $A=\tilde{\Pi}U\Pi$, where $\tilde{\Pi}$, $\Pi$ are operators that project onto the left and right singular vectors of $A$, respectively. By interleaving $U$ (and $U^\dagger$) with a series of projector-controlled phase-shift operators,
\begin{equation}
    \begin{split}
        U_{\vec{\phi}}=\tilde{\Pi}_{\phi_1}U\left[\prod_{n=1}^{(N-1)/2}\Pi_{\phi_{2n}}U^{\dagger}\tilde{\Pi}_{\phi_{2n+1}}U\right],
    \end{split}
    \label{eq:QSVT_operator_odd}
\end{equation}
one is able to apply, in this example, an odd (real-valued) polynomial transformation of degree $N$ to the singular values of $A$. That is,
\begin{equation}
p_N(A)\equiv Wp_N(\Sigma)V^\dagger=\tilde{\Pi}U_{\vec{\phi}}\Pi.
\end{equation}
As $\Sigma$ is a diagonal matrix, it is clear that $p_N(\Sigma)$ acts independently to transform each singular value $\{\sigma_k\}$ of $A$. An equivalent operator can be defined in the case that $p_N$ is an even polynomial. When $A$ is a Hermitian matrix, the transformation is applied to the eigenvalues of $A$; this procedure is aptly known as the quantum eigenvalue transformation (QET) \cite{Low:2016znh,Gilyen:2018khw}. In the standard formulation \cite{Gilyen:2018khw,PRXQuantum.2.040203}, the applied polynomial must satisfy three conditions:
\begin{enumerate}
    \item The polynomial must be real-valued.
    \item The polynomial must have definite parity.
    \item The polynomial must be bounded, $|p_N(x)|\leq1$ for $x\in[-1,1]$.
\end{enumerate}
The standard QSP, QET, and QSVT frameworks have each been generalized to remove the constraints of realness and definite parity \cite{Motlagh:2023oqc,Sunderhauf:2023peh}. The constraint of boundedness reflects the unitary nature of the the transformation and cannot be circumvented. In order to block encode the matrix $A$, it must be scaled so that its singular values lie within the interval $[0,1]$. Equivalently, a Hermitian matrix must be scaled so that its eigenvalues lie within $[-1,1]$. Hence the domain of interest for the polynomial $p_N(x)$ is $x\in[-1,1]$.

In general, the target function $p(x)$ is not identically a degree-$N$ polynomial but must be approximated as such. For these purposes, we require the maximum error between the target function $p(x)$ and its degree-$N$ polynomial approximation $p_N(x)$ on $[-1,1]$ to differ by no more than $\epsilon$; that is, 
\begin{equation}
    \|p(x)-p_N(x)\|\equiv\max_{x\in[-1,1]}|p(x)-p_N(x)|\leq \epsilon.
    \label{eq:max_err}
\end{equation}
Once a suitable polynomial approximation to the target function on the domain $[-1,1]$ has been obtained, efficient classical algorithms exist for the evaluation of the needed phase factors $\vec{\phi}=(\phi_1,\ldots,\vec{\phi_N})$ \cite{Chao:2020kcx,Dong:2020exv}. As evident from Eq. \eqref{eq:QSVT_operator_odd}, the query complexity of the block-encoding operator $U$ is linear in the degree $N$ of the applied polynomial. Therefore in practical calculations, one seeks to minimize the degree of the polynomial approximation while ensuring that the resulting error is below the required bound.  This is the primary motivation of the present work.

The Chebyshev polynomials (of the first kind) $T_n(x)$ form an orthogonal basis for functions on the interval $x\in[-1,1]$, allowing us to expand the target function exactly as
\begin{equation}
    p(x)=\sum_{n=0}^{\infty}a_nT_n(x).
\end{equation}
The coefficients $a_n$, known as \textit{Chebyshev projection} coefficients, are obtained by integrating the target function against the Chebyshev polynomials with a suitable weighting function,
\begin{equation}
    a_n=\frac{\delta_n}{\pi}\int_{-1}^1\frac{dx}{\sqrt{1-x^2}}p(x)T_{n}(x),
    \label{eq:cheby_coeff}
\end{equation}
where $\delta_0 = 1$ and $\delta_n=2$ for $n\geq 1$. In certain cases, the Chebyshev coefficients can be computed analytically, and when the summation is truncated at finite order,
\begin{equation}
    p_N(x)=\sum_{n=0}^Na_nT_n(x),
\end{equation} 
the residual error can be bounded by analyzing 
\begin{equation}
    \|p(x)-p_N(x)\|=\left\|\sum_{n=N+1}^{\infty}a_nT_n(x)\right\|.
\end{equation}
In general, however, it may not be possible to obtain closed-form expressions for these coefficients, and numerical evaluations of Eq. \eqref{eq:cheby_coeff} can be computationally difficult, requiring high precision at intermediate function evaluations --- due to the oscillatory nature of the Chebyshev polynomials, one often encounters very fine cancellations between large positive and negative contributions. Moreover when obtained numerically, these coefficients provide no information about the resulting error of the approximation, unless one performs an additional numerical optimization to compute $\|p(x)-p_N(x)\|$.

Although the Chebyshev projection exactly reproduces the target function in the limit $N\rightarrow \infty$, at finite $N$ there in general exist alternative choices of the coefficients $a_n$ that yield a smaller maximum error --- in the sense of Eq. \eqref{eq:max_err} --- than the Chebyshev projections. For fixed degree $N$, the choice of $a_n$ that minimizes the maximum error is known as the \textit{optimal polynomial}. 

One alternative method for determining a set of $a_n$ that approximates $p(x)$ on $[-1,1]$ is to sample the function $p(x)$ at the zeros of $T_N(x)$, which are conveniently given by $x_i=\cos[\pi(i+1/2)/N]$ for $i=0,\ldots,N-1$. From these $N$ samples, one then constructs a degree-$N$ polynomial by requiring that $p_N(x_i)=p(x_i)$ for all $0\leq i\leq N-1$. The resulting polynomial is called the \textit{Chebyshev interpolant} of degree $N$. In contrast to the Chebyshev projection, the Chebyshev interpolant is always efficient to construct and numerically stable \cite{10.1093/comjnl/15.2.156}. As before, however, knowledge of the Chebyshev interpolant itself does not reveal anything about the error of the approximation without additional numerical calculations.

 In certain instances, the optimal polynomial approximation to the relevant target function on $[-1,1]$ is known analytically \cite{Lin:2020acy}. For generic target functions, the Remez algorithm \cite{remez1934determination} can be used to iteratively construct the optimal polynomial of degree $N$ on $[-1,1]$. Unfortunately, the system of equations to be solved in this method is notoriously ill-conditioned and becomes very numerically challenging as the degree of the polynomial increases.
 
 Although Chebyshev projections and interpolants do not necessarily yield optimal polynomial approximations, in practice the resulting maximum error is very close to optimal (a notion that we will make precise in Theorem \ref{thm:lebesgue constants} below). When the target function $p(x)$ is analytic on $[-1,1]$, then both the Chebyshev projections and interpolants converge geometrically as $rM\rho^{-N}/(\rho-1)$ for some $M>0,\rho>1$ and where $r=2$ and $4$ for projections and interpolants, respectively (see Theorem 8.2 in \cite{doi:10.1137/1.9781611975949}). Here $\rho$ represents the sum of the semi-major and semi-minor axes of a particular ellipse, the so-called Bernstein ellipse, in the complex plane where the function $p(x)$ is analytically continued, and $M$ is the maximum value of $|p(z)|$ on this ellipse. In practice, larger values of $\rho$ --- corresponding to faster convergence in $N$ --- may correspond to significantly larger constant factors $M$, and there is no guarantee that the resulting error bounds are tight \cite{Tang:2023ywa}. Thus, although both approximations are provably close to optimal, it can be difficult to calculate a priori the minimal polynomial degree $N$ required to approximate $p(x)$ to within error $\epsilon$ using either Chebyshev projections or interpolants. 

\section{Nearly optimal polynomial approximations to even and odd step functions}
\begin{figure}
    \centering
    \includegraphics[scale=0.42]{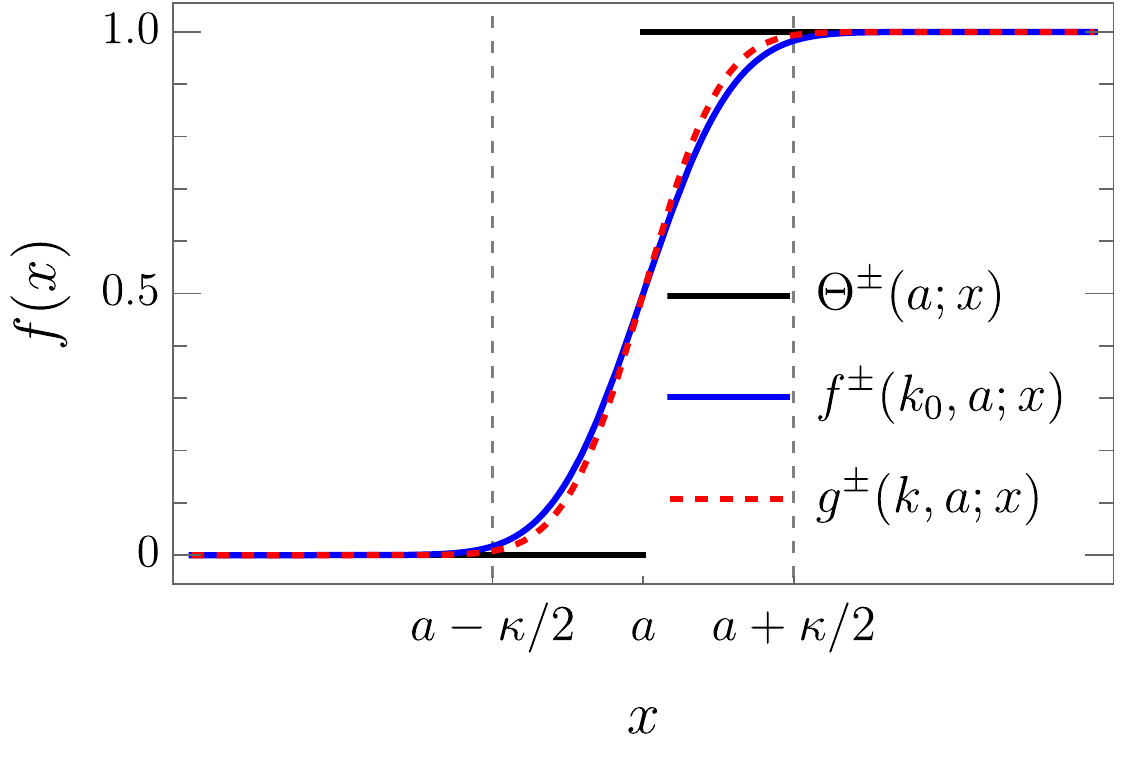}
    \caption{Sketch of the approximation of the discontinuous step functions $\Theta^\pm(a;x)$ by the smooth functions $f^\pm(k_0,a;x)$ and $g^\pm(k,a;x)$ for the same value of the scaling parameter $k_0=k$. Shaded vertical lines denote the region $[a-\kappa/2,a+\kappa/2]$ containing the discontinuity.}
    \label{fig:erf_plot}
\end{figure}

In QSP and QSVT, an important class of transformations are the even and odd step functions\footnote{We do not specify here the value of the function at the discontinuity, as it has no practical consequence for our work. But for completeness, one can adopt $\Theta^\pm(a;x)=\pm1/2$ when $x=\pm a$.} 
\begin{equation}
    \Theta^\pm(a;x)\equiv \left\{\begin{array}{cc}
       0,  &  |x|<a, \\
       1, & x>a, \\
       \pm1, & x<-a,
    \end{array}\right.
\end{equation}
where $0\leq a<1$ determines the location of the discontinuities. These functions arise, for example, when performing spectral filtering, linear amplitude amplification, and quantum phase estimation (QPE). Even in matrix-inverse problems where one has an analytic form for the optimal polynomial approximation to $1/x$ for $|x|\in[\kappa,1]$ \cite{privalov,Sunderhauf:2025lbc}, it is still necessary to multiply this polynomial by an additional window function, effectively $\Theta^+(a;x)$, in order to guarantee that the polynomial is bounded by 1 for $|x|<\kappa$. 

In the literature \cite{low_2017,PRXQuantum.2.040203}, the step functions $\Theta^+(a;x)$ are commonly approximated using either even or odd combinations of shifted error functions,
\begin{equation}
\begin{split}
        f^+(k_0,a;x)&\equiv 1-\frac{1}{2}\{\mathrm{erf}[k_0(a+x)]\\
        &~~~~~~~~~~~~-\mathrm{erf}[k_0(x-a)]\},\\
        f^-(k_0,a;x)&\equiv \frac{1}{2}\{\mathrm{erf}[k_0(a+x)]\\
        &~~~~~~~~+\mathrm{erf}[k_0(x-a)]\},\\
\end{split}
\label{eq:fplusminus}
\end{equation}
where the scaling parameter $k_0>0$ controls the sharpness of the transition. As shown in Fig. \ref{fig:erf_plot}, the function $f^\pm(k_0,a;x)$ approximates $\Theta^\pm(a;x)$ on the region $x\in D_\kappa(a)$,
where
\begin{equation}
 D_\kappa(a)\equiv \left\{x : |x|\in[0,a-\kappa/2]\cup[a+\kappa/2,1]\right\},  
 \label{eq:D_kappa}
\end{equation}
that is, everywhere on $[-1,1]$ except a window of size $\kappa$ centered on each discontinuity of the step function. As we show in detail below, larger values of $k_0$ correspond to smaller windows $\kappa$ and overall better approximations of the step functions. 

To this end, we distinguish two sources on error when approximating $\Theta^\pm(a;x)$ with polynomial approximations to smooth functions: The first is the error that arises when approximating the discontinuous function $\Theta^\pm(a;x)$ by the smooth functions $f^\pm(k_0,a;x)$ [and later $g^\pm(k,a;x)$]. As discussed above, this error is only considered on a region $D_\kappa(a)$ that excludes a small window around each discontinuity. Throughout this work, we employ $\delta$ to represent this type of error; for example,
\begin{equation}
    \begin{split}
    \delta&=\max_{x\in D_\kappa(a)}|f^\pm(k_0,a;x)-\Theta^\pm(a;x)|\\
    &\equiv\|f^\pm(k_0,a;x)-\Theta^\pm(a;x)\|_{D_\kappa(a)}.
    \end{split}
\end{equation}
The second source of error is that arising from the polynomial approximation of the smooth function; we will consistently employ the symbol $\epsilon$ to represent this type of error, as in \begin{equation}
    \epsilon=\|f^\pm_N(k_0,a;x)-f^\pm(k_0,a;x)\|.
\end{equation} 
Away from the discontinuities, the maximum total deviation between the step function and its polynomial approximation is then bounded by
\begin{equation}
    \|f_N^\pm(k_0,a;x)-\Theta^\pm(a;x)\|_{D_\kappa(a)}\leq \delta+\epsilon.
\end{equation}

Existing polynomial approximations of the shifted error functions are derived from the Jacobi-Anger expansion of the exponential decay function \cite{low_2017},
\begin{equation}
    e^{-\beta(x+1)}=e^{-\beta}\left(I_0(\beta)+2\sum_{n=1}^{\infty}I_n(\beta)T_n(-x)\right),
\end{equation}
where $I_n(\beta)$ is a modified Bessel function of the first kind. To produce a polynomial of degree $N$, the expansion is truncated at order $N$,
\begin{equation}
    p^\mathrm{exp}_N(\beta;x)\equiv e^{-\beta}\left(I_0(\beta)+2\sum_{n=1}^{N}I_n(\beta)T_n(-x)\right),
\end{equation}
and the maximum error of this approximation is precisely
\begin{equation}
  \|e^{-\beta(x+1)}-p_N^\mathrm{exp}(\beta;x)\|=2e^{-\beta}\sum_{n=N+1}^{\infty}|I_n(\beta)|.
\end{equation}
Starting from $p^\mathrm{exp}_N(\beta;x)$, one obtains a polynomial approximation of $\mathrm{erf}(k_0x)$ by first considering $p_N^\mathrm{exp}(k_0^2/2;2x^2-1)$, which approximates the Gaussian function $e^{-k_0^2x^2}$ on $[-1,1]$. Next, using the definition of the error function
\begin{equation}
    \mathrm{erf}(k_0x)=\frac{2k_0}{\sqrt{\pi}}\int_0^{x}dt~e^{-k_0^2t^2},
\end{equation}
one integrates the polynomial approximation of $e^{-k_0^2x^2}$, term-by-term with respect to $x$ to obtain
\begin{equation}
\begin{split}
    &p^\mathrm{erf}_N(k_0;x)=\frac{2k_0e^{-k_0^2/2}}{\sqrt{\pi}}\Bigg\{I_0(k_0^2/2)x\\
    &~~~+\sum_{j=1}^{N}I_j(k_0^2/2)(-1)^j\left[\frac{T_{2j+1}(x)}{2j+1}-\frac{T_{2j-1}(x)}{2j-1}\right]\Bigg\}.
    \end{split}
    \label{eq:their_erf_poly}
\end{equation}

To extend this polynomial approximation to an approximation of the shifted error function, one cannot simply take $p_N^\mathrm{erf}(k_0;x\pm a)$ because the argument of the Chebyshev polynomials must remain in the domain $[-1,1]$. To ensure this, one considers the polynomial approximations
\begin{equation}
    \begin{split}
        f_N^+(k_0,a;x)&=1-\frac{1}{2}\left[p_N^\mathrm{erf}(k';x^+)+p_N^\mathrm{erf}(k';x^-)\right],\\
        f_N^-(k_0,a;x)&=\frac{1}{2}\left[p_N^\mathrm{erf}(k';x^+)-p_N^\mathrm{erf}(k';x^-)\right],
    \end{split}
\end{equation}
where 
\begin{equation}
\begin{split}
k'&\equiv (1+a)k_0,\\
x^\pm&\equiv (x\pm a)/(1+a).
\end{split}
\end{equation}
Thus, in order to ensure that the argument of the Chebyshev polynomials remains within the domain $[-1,1]$, we must increase the value of the scaling parameter by a factor of $1+a$.\footnote{In Ref. \cite{low_2017}, it was originally suggested that for all $0<a<1$, one can take the most conservative rescaling $x^\pm=(a\pm x)/2$ and $k'=2k_0$. The primary motivation in that work was to demonstrate asymptotic scaling, which is not affected by this choice, but in practice it is advantageous to take the limited rescaling by $1+a$ that we consider in this work.}

As we demonstrate numerically in Sec. \ref{sec:sign_func}, $f_N^-(k_0,0;x)$ is a nearly optimal polynomial approximation of $\mathrm{erf}(k_0x)$ with maximum errors that are competitive with Chebyshev interpolants of the same order. Unfortunately, the quality of the approximation degrades rapidly when the argument of the error function is shifted --- for example to $\mathrm{erf}[k_0(x\pm a)]$ as in Eq. \eqref{eq:fplusminus} --- and the polynomials $f_N^\pm(k_0,a;x)$ become far from optimal for $a>0$. In all cases, the degree $N$ of the polynomial required to achieve error bounded by $\epsilon$ scales as $N\approx  O[k_0\log^{1/2}(1/\epsilon)]$. The polynomials that we introduce in this work exhibit the same scaling with improved constant factors when $a>0$. 

In addition to the sub-optimal performance of $f_N^\pm(k,a;x)$ when $a>0$, the standard approach suffers two additional disadvantages: (1) The theoretical upper bound on the error $\epsilon$ derived in Ref. \cite{low_2017} is extremely loose and is not useful in practice for estimating the polynomial degree $N$ required to achieve target error $\epsilon$. (2) At large $k_0$, numerical evaluation of the Chebyshev coefficients in Eq. \eqref{eq:their_erf_poly} requires high precision in order to resolve large cancellations between the exponential suppression of $e^{-k_0^2/2}$ and the exponential growth of $I_j(k_0^2/2)$.

Our approach is instead based on approximating the step functions $\Theta^\pm(a;x)$ with the smooth functions
\begin{equation}
\begin{split}
        g^+(k,a;x)&\equiv 1-\frac{1}{2}\{\mathrm{erf}[k(\arcsin x+\arcsin a)] \\
        &~~~~~~~~~- \mathrm{erf}[k(\arcsin x - \arcsin a)]\},\\
        g^-(k,a;x)&\equiv \frac{1}{2}\{\mathrm{erf}[k(\arcsin x+\arcsin a)] \\
        &~~~~+ \mathrm{erf}[k(\arcsin x - \arcsin a)]\},
\end{split}
\label{eq:g_plus_minus}
\end{equation}
which correspond to the replacements $a\rightarrow \arcsin a$, $x\rightarrow \arcsin x$ in Eq. \eqref{eq:fplusminus}. Upon first inspection, the $g^\pm(k,a;x)$ functions appear to be needlessly complicated in comparison to $f^\pm(k,a;x)$, but as we demonstrate throughout this work, the former have many desirable properties not possessed by the latter. In particular:

\begin{enumerate}
    \item For $a>0$, the polynomial approximation of $g^\pm(k,a;x)$ does not require any rescaling of the domain. In fact, as $a$ increases (for fixed $\kappa>0$), $g^\pm(k,a;x)$ requires a smaller value of $k$ in order to approximate $\Theta^\pm(a;x)$ to within error $\delta$ on $D_\kappa(a)$. This property is illustrated in Fig. \ref{fig:erf_plot}, where $\mathrm{erf}[k_0(x-a)]$ and $\mathrm{erf}[k(\arcsin x-\arcsin a)]$ are plotted for the same scaling parameter $k_0=k$. The latter function shows a significantly steeper transition between 0 and 1.
    \item The Chebyshev projection coefficients of $g^\pm (k,a;x)$ can be computed in closed form.
    \item Although the exact Chebyshev coefficients are complicated, for sufficiently large $k$, $g^\pm(k,a;x)$ can be approximated to very high accuracy by related polynomials with much simpler Chebyshev coefficients. 
    \item For sufficiently large $k$, $g_N^\pm(k,a;x)$ are nearly optimal polynomial approximations of $g^\pm(k,a;x)$ for all $0\leq a<1$.
\end{enumerate}

The polynomial approximations that we consider are exceptionally simple. For odd $N\geq 1$, we define the odd polynomial of degree $N$ as
\begin{equation}
    \begin{split}
            g^-_N(k&,a;x)\equiv\frac{4}{\pi}\sqrt{1-a^2}\\
            &\times\sum_{n=0}^{(N-1)/2} \frac{U_{2n}(a)}{2n+1}e^{-(2n+1)^2/4k^2}T_{2n+1}(x),
    \end{split}
        \label{eq:gpoly_minus}
\end{equation}
where $U_n(x)$ is the Chebyshev polynomial of the second kind. For even $N\geq 0$, we define the even polynomial of degree $N$ as
\begin{equation}
    \begin{split}
            g^+_N(k&,a;x)\equiv\frac{2}{\pi}\arccos a\\&+\frac{4}{\pi}\sqrt{1-a^2}\sum_{n=1}^{N/2} \frac{U_{2n-1}(a)}{2n}e^{-n^2/k^2}T_{2n}(x).
    \end{split}
        \label{eq:gpoly_plus}
\end{equation}
The Chebyshev coefficients that define $g_N^\pm(k,a;x)$ are damped by an exponential factor for $n\gtrsim k$. If one removes this exponential damping, the resulting polynomials are precisely the Chebyshev projections of $\Theta^\pm(k;x)$. While it may seem tempting to work directly with these expressions, polynomial approximations of discontinuous functions suffer from the so-called Gibbs phenomenon \cite{doi:10.1137/S0036144596301390}, where significant errors cluster near the points of discontinuity. We choose instead to take the two-step approach of first approximating $\Theta^\pm (a;x)$ by the smooth functions $g^\pm(k,a;x)$ and then approximating these smooth functions by the polynomials in Eqs. \eqref{eq:g_plus_minus} and \eqref{eq:gpoly_plus}. The polynomials $g_N^\pm (k,a;x)$ can therefore be viewed as the Chebyshev projections of $\Theta^\pm(a;x)$ with an additional Gaussian filter function to smooth the Gibbs oscillations at large $n$. 

We now state our main results in a series of theorems: Theorem \ref{thm:g_minus_fit_proof} demonstrates the ability of the smooth function $g^-(k,a;x)$ to approximate the discontinuous function $\Theta^-(a;x)$ on $[-1,1]$ excluding a small neighborhood around each discontinuity. Theorem \ref{thm:gtilde_odd} proves that the polynomial $g_N^-(k,a;x)$ defined in Eq. \eqref{eq:gpoly_minus} converges to $g^-(k,a;x)$ as $N\rightarrow \infty$, up to terms that are exponentially suppressed at sufficiently large $k$. Theorem \ref{thm:gpoly_odd} provides a theoretical upper bound on the error resulting from truncating the polynomial approximation $g_N^-(k,a;x)$ at order $N$. Theorems \ref{thm:g_plus_fit_proof}, \ref{thm:gtilde_even}, and \ref{thm:gpoly_even} provide corresponding results for the even function $g^+(k,a;x)$ and its polynomial approximation $g_N^+(k,a;x)$. The proofs of these theorems are relegated to Appendices \ref{app:g_minus_fit_proof}--\ref{app:gpoly_even}.

\begin{theorem}
\label{thm:g_minus_fit_proof}
Let $0<\kappa<1$ and $0\leq a<1-\kappa/2$. Define 
\begin{equation}
\begin{split}
    \eta&=\arcsin(a)-\arcsin(a-\kappa/2),\\
    \nu&=\arcsin(a+\kappa/2)-\arcsin(a),
\end{split}
\label{eq:eta_nu}
\end{equation}
and let $g^-(k,a;x)$ as in Eq. \eqref{eq:g_plus_minus}. Then for all $\delta>0$, choosing 
 \begin{equation}
        k=\frac{1}{\sqrt{2}}\max\left[\frac{1}{\eta}\sqrt{W\left(\frac{1}{2\pi\delta^2}\right)},\frac{1}{\nu}\sqrt{W\left(\frac{2}{\pi\delta^2}\right)}~\right],
        \label{eq:k_req_gminus}
    \end{equation}
    where $W(x)$ is the Lambert-$W$ (product-log) function guarantees that
\begin{equation}
    \|\Theta^-(a;x)-g^-(k,a;x)\|_{D_\kappa(a)}\leq \delta.
\end{equation}
\end{theorem}

\begin{theorem}
\label{thm:gtilde_odd}
    Let $0\leq a<1$ and let $ \tilde{g}^-(k,a;x)=\lim_{N\rightarrow \infty}g^-_N(k,a;x)$. Then for all $\tilde{\epsilon}>0$, choosing 
\begin{equation}
    k=\frac{1}{\sqrt{2}\arccos a}\sqrt{W\left(\frac{9}{2\pi\tilde{\epsilon}^2}\right)}
\end{equation}
guarantees that 
\begin{equation}
   \|\tilde{g}^-(k,a;x)-g^-(k,a;x)\|\leq \tilde{\epsilon}.
\end{equation}
\end{theorem}

\begin{theorem}
\label{thm:gpoly_odd}
 Let $0\leq a<1$. Let $g^-_N(k,a;x)$ as in Eq. \eqref{eq:g_plus_minus} and let $\tilde{g}^-(k;a;x)=\lim_{N\rightarrow \infty}g_N^-(k,a;x)$. Then for all $\epsilon>0$, choosing 
\begin{equation}
    N=2\ceil*{k\sqrt{W(1/(\pi\epsilon))}}+1,
\end{equation}
guarantees that
\begin{equation}
     \|\tilde{g}^-(k,a;x)-g^-_N(k,a;x)\|\leq\epsilon.
\end{equation}

\end{theorem}

\begin{theorem}
\label{thm:g_plus_fit_proof}
Let $\kappa>0$ and $0\leq a<1-\kappa/2$. Let $\eta,\nu$ as in Eq. \eqref{eq:eta_nu}, and let $g^+(k,a;x)$ as in Eq. \eqref{eq:g_plus_minus}. Then for all $\delta>0$, choosing 
 \begin{equation}
        k=\frac{1}{\sqrt{2}}\max\left[\frac{1}{\eta}\sqrt{W\left(\frac{2}{\pi\delta^2}\right)},\frac{1}{\nu}\sqrt{W\left(\frac{1}{2\pi\delta^2}\right)}~\right],
        \label{eq:k_req_gplus}
    \end{equation}
  guarantees that
\begin{equation}
    \|\Theta^+(a;x)-g^+(k,a;x)\|_{D_\kappa(a)}\leq \delta.
\end{equation}
\end{theorem}

\begin{theorem}
\label{thm:gtilde_even}
    Let $0\leq a<1$. Let $g_N^+(k,a;x)$ as in Eq. \eqref{eq:g_plus_minus} and let $\tilde{g}^+(k,a;x)=\lim_{N\rightarrow\infty}g_N^+(k,a;x)$. Then for all $\tilde{\epsilon}>0$, choosing 
    \begin{equation}
    k\geq \frac{1}{\sqrt{2}\arccos a}\sqrt{W\left(\frac{9+4\sqrt{2}}{2\pi\tilde{\epsilon}^2}\right)}
    \end{equation}
    guarantees that 
    \begin{equation}
        \|\tilde{g}^+(k,a;x)-g^+(k,a;x)\|\leq\tilde{\epsilon}
    \end{equation}
\end{theorem}

\begin{theorem}
\label{thm:gpoly_even}
 Let $0\leq a<1$. Let $g^+_N(k,a;x)$ as in Eq. \eqref{eq:g_plus_minus} and let $\tilde{g}^+(k;a;x)=\lim_{N\rightarrow \infty}g_N^+(k,a;x)$. Then for all $\epsilon>0$, choosing 
\begin{equation}
    N=2\ceil*{k\sqrt{W(1/(\pi\epsilon))}},
    \label{eq:N_eps_plus}
\end{equation}
guarantees that the even polynomial $g^+_N(k,a;x)$ satisfies
\begin{equation}
      \|\tilde{g}^+(k,a;x)-g^+_N(k,a;x)\|\leq\epsilon.
\end{equation}
\end{theorem}

\subsection{Remarks}
\label{sec:remarks}
In Theorems \ref{thm:g_minus_fit_proof} and \ref{thm:g_plus_fit_proof}, we have given tight theoretical bounds on the error between $g^\pm(k,a;x)$ and $\Theta^\pm(a;x)$ in terms of the Lambert-$W$ function. Very fast and accurate numerical evaluations of the principal branch of the Lambert-$W$ function without the need for extended precision \cite{num_W_eval}. Therefore Eqs. \eqref{eq:k_req_gminus} and \eqref{eq:k_req_gplus} can be used in practice. Nonetheless, to understand the scaling relations between $k$, $a$, $\kappa$, and $\delta$ more intuitively we can expand for small $\kappa$, 
\begin{equation}
    \eta\approx\nu \approx \frac{\kappa}{2\sqrt{1-a^2}}.
\end{equation}
The requisite value of the scaling parameter $k$ [for both $g^-(k,a;x)$ and $g^+(k,a;x)$] is then given by 
\begin{equation}
    \begin{split}
    k&\approx \sqrt{1-a^2}\frac{\sqrt{2}}{\kappa}\sqrt{W\left(\frac{2}{\pi\delta^2}\right)}\\
    &=\sqrt{1-a^2}k_0,
    \end{split}
    \label{eq:k_approx}
\end{equation}
where 
\begin{equation}
    k_0\equiv \frac{\sqrt{2}}{\kappa}\sqrt{W\left(\frac{2}{\pi\delta^2}\right)},
    \label{eq:k0_scale}
\end{equation}
is the scaling parameter required for the function $f^\pm(k_0,a;x)$ \cite{low_2017} to approximate $\Theta^\pm(a;x)$ to within error $\delta$ on $D_{\kappa(a)}$, which is independent of $a$. 

Finally, we can utilize the bound $W(x)\leq \log(x)$ for $x>e$ to express the above bounds in simpler terms. Thus, whenever $0<\delta< \sqrt{2/\pi e}$, it suffices to choose 
\begin{equation}
    k=\sqrt{1-a^2}\frac{\sqrt{2}}{\kappa}\log^{1/2}\left(\frac{2}{\pi\delta^2}\right),
\end{equation}
though this bound is often quite loose compared to Eqs. \eqref{eq:k_req_gminus} and \eqref{eq:k_req_gplus}.

Figure \ref{fig:k_scaling} compares the value of the scaling parameter required for the smooth functions $f^\pm (k_0,a;x)$ and $g^\pm(k,a;x)$ to approximate $\Theta^\pm(a;x)$ to within maximum deviation $\delta=10^{-3}$ on the region $D_\kappa(a)$ for $\kappa=0.05$. In the figure, we have restored the exact expressions in Eqs. \eqref{eq:k_req_gminus} and \eqref{eq:k_req_gplus}, so there is a small difference in the required $k$ value for $g^- (k,a;x)$ compared to $g^+ (k,a;x)$ for $a>0$. In both cases, the required $k$ decreases steadily --- roughly like $\sqrt{1-a^2}$ --- as the shift parameter $a$ increases toward the maximum allowed value of $a=1-\kappa/2=0.975$. In contrast, the value of $k_0$ --- determined by Eq. \eqref{eq:k0_scale} --- required for $f^\pm(k_0,a;x)$ to approximate $\Theta^\pm(a;x)$ to within error $\delta$ on $D_\kappa (a)$ remains constant with $a$. However, it is the value of $k'=(1+a)k_0$ that determines the order of the polynomial approximation required for $f^\pm_N(k_0,a;x)$ to approximate $f^\pm(k_0,a;x)$ to within error $\epsilon$. 

As demonstrated for $g_N^-(k,a;x)$ in Theorem \ref{thm:gpoly_odd} and for $g^+_N(k,a;x)$ in Theorem \ref{thm:gpoly_even}, the degree $N$ of the polynomial required to achieve error $\epsilon$ grows linearly with the scaling parameter $k$. Similarly, for $f_N^\pm(k_0,a;x)$, $N$ grows linearly with $k'=(1+a)k_0$. Therefore to achieve the same total approximation error $\delta+\epsilon$ as the degree-$N$ polynomial $g_N^\pm(k,a;x)$, one expects to need a degree-$M$ polynomial $f_M^\pm(k_0,a;x)$ with $M/N\approx (1+a)/\sqrt{1-a^2}$.  While the polynomials introduced in this work do not improve upon the asymptotic scaling of existing approximations\footnote{As our polynomials are nearly optimal, no such improvement is possible.}, the constant-factor improvements can be significant, especially for implementations on near-term quantum devices. When $a=1/\sqrt{2}\approx 0.71$ --- the value required to implement QPE with QSVT --- the above analysis implies that $f_M^+(k_0,a;x)$ requires a polynomial of degree $M\approx 2.4 \times N$ to achieve the same total error as $g_N^+(k,a;x)$. Thus, for a fixed target error, the use of our polynomials reduces the required number of applications of the block-encoding matrix by more than a factor of two in this case. This example is explored further in Sec. \ref{sec:QPE} below.

\begin{figure}
    \centering
    \includegraphics[scale=0.54]{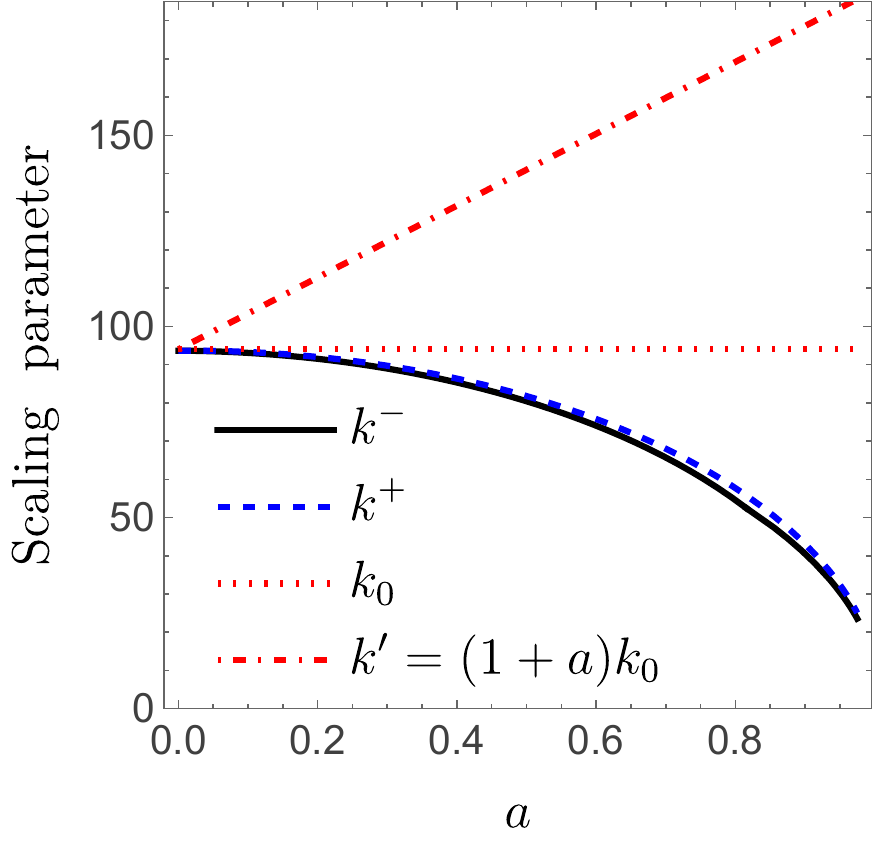}
    \caption{Value of the scaling parameter required to approximate $\Theta^\pm (a;x)$ to within error $\delta=10^{-3}$ on the region $D_{\kappa}(a)$ for $\kappa=0.05$. The solid black (dashed blue) curve shows the value of $k$ required for the function $g^-(k,a;x)$ [$g^+(k,a;x)$] to approximate $\Theta^-(a;x)$ [$\Theta^+(a;x)]$ in accordance with Theorem \ref{thm:g_minus_fit_proof} (Theorem \ref{thm:g_plus_fit_proof}). The red dotted curve shows the value of $k_0$ required for $f^\pm(k,a;x)$ to approximate $\Theta^\pm(a;x)$ to within the desired tolerance, while the dash-dotted curve shows the corresponding value $k'=(1+a)k_0$ that enters the shifted polynomial approximation in Eq. \eqref{eq:their_erf_poly}.} 
    \label{fig:k_scaling}
\end{figure}

In addition to the error $\delta$, arising from approximating $\Theta^\pm(a;x)$ by smooth functions, and the error $\epsilon$, arising from approximating the smooth functions by polynomials, there is a third type of error, $\tilde{\epsilon}$, that is relevant for the polynomials $g^\pm(k,a;x)$. In Theorems \ref{thm:gtilde_odd} and \ref{thm:gtilde_even}, $\tilde{\epsilon}$ measures the maximum deviation between the smooth functions $g^\pm(k,a;x)$ and the infinite series $\lim_{N\rightarrow \infty}g^\pm_N(k,a;x)$. The proofs of these theorems (see Appendices \ref{app:gtilde_odd} and \ref{app:gtilde_even}) reveal that $\tilde{\epsilon}$ decreases exponentially with increasing $k$, just as does the error $\delta$, albeit with different prefactors. As such, Theorems \ref{thm:g_minus_fit_proof} and \ref{thm:gtilde_odd} (equivalently, Theorems \ref{thm:g_plus_fit_proof} and \ref{thm:gtilde_even}) both suggest a minimal value of $k$ required to achieve respective errors $\delta$ and $\tilde{\epsilon}$. If $\delta= \tilde{\epsilon}$, then the value of $k$ required to achieve error $\delta$ as dictated by Theorem \ref{thm:g_minus_fit_proof} (Theorem \ref{thm:g_plus_fit_proof}) is always greater than or equal to the value required to achieve error $\tilde{\epsilon}$ in Theorem \ref{thm:gtilde_odd} (Theorem \ref{thm:gtilde_even}). In practice, if one adopts the value of $k$ required to achieve error $\delta$, then the corresponding error $\tilde{\epsilon}$ is often completely negligible.

Independent of the shift parameter $a$, the error $\epsilon$ of the polynomial approximation $g_N^\pm(k,a;x)$ is bounded above by
\begin{equation}
    \epsilon\leq \frac{4}{\pi}\frac{k^2}{N^2}e^{-N^2/4k^2}.
    \label{eq:poly_error}
\end{equation}
As demonstrated in Appendices \ref{app:gpoly_odd} and \ref{app:gpoly_even}, this result is obtained by assuming the maximal error for terms depending on the value of the shift parameter, $|\sqrt{1-a^2}U_{2n}(a)|\leq 1$ for $0\leq a<1$. This bound is saturated when $a=0$, and therefore we expect the bound in Eq. \eqref{eq:poly_error} to be relatively tight for this value of the shift parameter. Likewise, the bound in Eq. \eqref{eq:poly_error} may weaken significantly for $a>0$. These expectations are consistent with what we observe in the numerical examples in Sec. \ref{sec:applications} below, though even when the error bound weakens, it is still quite tight.

\subsection{The notion of ``nearly optimal''}
The following theorem formalizes the notion that polynomial approximations based on Chebyshev projections and interpolants are \textit{nearly optimal}:
\begin{theorem}
    \label{thm:lebesgue constants}
    Let $p(x)$ be a continuous function on $[-1,1]$. Let $p_N(x)$ be the order-$N$ Chebyshev interpolant to $p(x)$, let $q_N(x)$ be the order-$N$ Chebyshev projection of $p(x)$, and let $p_N^*(x)$
 be the optimal order-$N$ polynomial approximation to $p(x)$. Then 
 \begin{equation}
 \begin{split}
     \|f(x)-p_N(x)\|&\leq (\Lambda_N^{(1)}+1)\|f(x)-p^*_N(x)\|,\\
    \|f(x)-q_N(x)\|&\leq (\Lambda_N^{(2)}+1)\|f(x)-p^*_N(x)\|,
    \end{split}
 \end{equation}
 where $\Lambda^{(1)}_N$, $\Lambda^{(2)}_N$ are Lebesgue constants bounded above by
 \begin{equation}
 \begin{split}
    \Lambda^{(1)}_N&\leq \frac{2}{\pi}\log(N+1)+1,\\
    \Lambda^{(2)}_N&\leq \frac{4}{\pi^2}\log(N+1)+3.
    \end{split}
 \end{equation}
 \end{theorem}
For proof of these statements, see Theorems 15.1, 15.2, and 15.3 in \cite{doi:10.1137/1.9781611975949}. The growth of the Lebesgue constants that bound the difference between the error of the Chebyshev projections/interpolants and the optimal error is extremely slow: Even for $N=10^5$, $\Lambda^{(1)}_N+1\approx 9.3$, so that the optimal polynomial represents not even an order of magnitude improvement in accuracy over the simple interpolant. Thus the maximum errors of the Chebyshev interpolants and projections are guaranteed to differ from the optimal polynomial by a multiplicative factor that grows extremely slowly with the polynomial order, and in this sense we consider them to be \textit{nearly optimal}.

Theorems \ref{thm:gtilde_odd} and \ref{thm:gtilde_even} demonstrate that the polynomials $g_N^\pm(k,a;x)$ do not exactly converge to $g^\pm(k,a;x)$ as $N\rightarrow \infty$, but they do approximate these functions up to terms that are exponentially suppressed at large $k$. Thus, while $g_N^\pm(k,a;x)$ cannot correspond to the exact Chebyshev projection of $g^\pm(k,a;x)$, we expect that the two sets of Chebyshev coefficients must be very close for sufficiently large $k$. Indeed, we have the following theorems:
\begin{theorem}
\label{thm:g_odd_proj}
    Let $0\leq a<1$ and let
\begin{equation}
\tilde{a}_{2n+1}=\frac{4}{\pi}\sqrt{1-a^2}\frac{U_{2n}(a)}{2n+1}e^{-(2n+1)^2/4k^2},
\label{eq:an_odd}
\end{equation}
be the Chebyshev coefficients of $\tilde{g}^-(k,a;x)=\lim_{N\rightarrow \infty}g^-(k,a;x)$. Consider the exact Chebyshev projections of $g^-(k,a;x)$,
\begin{equation}
\begin{split}
        a_{2n+1}&=\frac{2}{\pi} \int_{-1}^1 \frac{dx}{\sqrt{1-x^2}}g^-(k,a;x)T_{2n+1}(x)\\
        &= \tilde{a}_{2n+1}+\Delta a_{2n+1}.
\end{split}
\end{equation}
Then for all $\tilde{\epsilon} >0$, choosing
\begin{equation}
    k=\frac{1}{\sqrt{2}\arccos a}\sqrt{W\left(\frac{32}{\pi^3\tilde{\epsilon}^2}\right)}
\end{equation}
guarantees that $|\Delta a_{2n+1}|\leq \tilde{\epsilon}$ for all $n\geq 0$.
\end{theorem}
\begin{theorem}
\label{thm:g_even_proj}
    Let $0\leq a<1$ and define 
    \begin{equation}
    \begin{split}
        \tilde{a}_0&=\frac{2\arccos a}{\pi},\\
        \tilde{a}_{2n}&= \frac{4}{\pi}\sqrt{1-a^2}\frac{U_{2n-1}(a)}{2n} e^{-n^2/k^2},
    \end{split}
    \label{eq:an_even}
    \end{equation}
    be the Chebyshev coefficients of $\tilde{g}^+(k,a;x)=\lim_{N\rightarrow\infty}g_N^+(k,a;x)$. Consider the exact Chebyshev projections 
\begin{equation}
\begin{split}
        a_{2n}&=\frac{\delta_n}{\pi}\int_{-1}^1\frac{dx}{\sqrt{1-x^2}}g^+(k,a;x)T_{2n}(x)\\
        &=\tilde{a}_{2n}+\Delta a_{2n}.
\end{split}
\end{equation}
Then for all $\tilde{\epsilon}>0$, choosing 
\begin{equation}
    k=\frac{1}{\sqrt{2}\arccos a}\sqrt{W\left(\frac{8}{\pi\tilde{\epsilon}^2}\right)},
\end{equation}
guarantees that $|\Delta a_{2n}|\leq \tilde{\epsilon}$ for all $n\geq 0$.
\end{theorem}

The proofs of Theorems \ref{thm:g_odd_proj} and \ref{thm:g_even_proj} are given in Appendices \ref{app:g_odd_proj} and \ref{app:g_even_proj}, respectively.

This approximation converges extremely rapidly with $k$. For example when $a=1/2$ and $k=10$, already $\epsilon\lesssim  10^{-49}$. Thus for large $k$, $a_{2n+1}\approx \tilde{a}_{2n+1}$ to very high precision, and the polynomial described the $a_{2n+1}$ coefficients effectively corresponds to the Chebyshev projection. Therefore the polynomial furnished by the $a_{2n+1}$ is nearly optimal in the sense of Theorem \ref{thm:lebesgue constants}.

\section{Applications}
\label{sec:applications}
In this section, we demonstrate the utility of our polynomial approximations by applying them to a series of relevant problems within the QSVT framework. In each case, we will compare the accuracy our polynomials $g_N^\pm(k,a;x)$ to that of the $f_N^\pm(k,a;x)$ polynomials first introduced in Ref. \cite{low_2017}. We will also compare against the corresponding Chebyshev interpolants, demonstrating numerically that the $g^\pm_N(k,a;x)$ polynomials are nearly optimal. From the (numerically determined) errors on the Chebyshev interpolants, we also place a rigorous lower bound on the error of the optimal polynomial in accordance with Theorem \ref{thm:lebesgue constants}, further quantifying the notion that our approximations are close to optimal.

\subsection{Sign function}
\label{sec:sign_func}
The function 
\begin{equation}
    \mathrm{sign}(x)=\Theta^-(0;x)=\left\{\begin{array}{lr}
        -1, & x<0, \\
       +1,  & x>0,
    \end{array}\right.
\end{equation}
maps all input values to 1 for $x>0$, and can therefore be used to perform oblivious amplitude amplification \cite{PRXQuantum.2.040203}. We first approximate $\mathrm{sign}(x)$ by the smooth functions
\begin{equation}
\begin{split}
        f^-(k_0,0;x)&=\mathrm{erf}(k_0x),\\
        g^-(k,0;x)&=\mathrm{erf}(k\arcsin x).
\end{split}
\end{equation}

\begin{figure}
    \centering
    \includegraphics[scale=0.42]{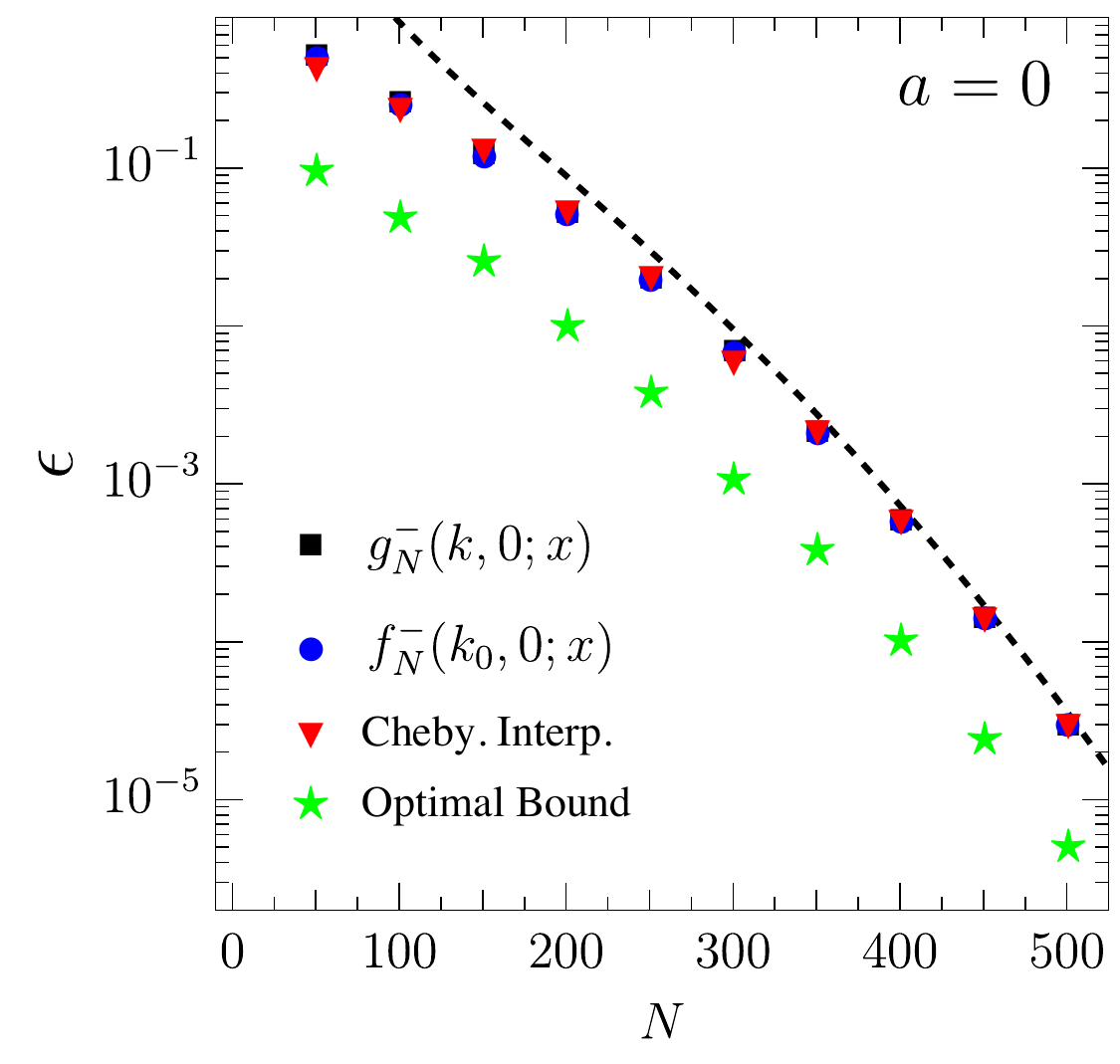}
    \caption{Comparison of the error $\epsilon$ for different polynomial approximations of the sign function $\mathrm{sign}(x)=\Theta^-(0;x)$. Black squares correspond to the polynomial $g_N^-(k,0;x)$ introduced in Eq. \eqref{eq:gpoly_minus}; blue circles denote the polynomial $f_N^-(k_0,0;x)$ from Eq. \eqref{eq:their_erf_poly}; red triangles correspond to the Chebyshev interpolant of $g^-(k,0;x)$; green stars represent a theoretical lower bound on the error of the optimal polynomial approximation to $g^-(k,0;x)$ (see Theorem \ref{thm:lebesgue constants}). The dashed black line denotes the theoretical upper bound on the error of $g_N^-(k,0;x)$ from Eq. \eqref{eq:poly_error}. The value of the scaling parameter $k\approx k_0\approx 94$ was determined by requiring that $\delta<10^{-3}$ on $D_{\kappa(a)}$ with $\kappa=0.05$.}
    \label{fig:poly_sign_compare}
\end{figure}

Based on Eq. \eqref{eq:k_approx}, for $a=0$ and small $\kappa$, we expect the scaling parameters required by $g^-(k,a;x)$ and $f^-(k_0,a;x)$ to be approximately equal, $k\approx k_0$. Therefore, the corresponding polynomials $g^-_N(k,0;x)$ and $f^-(k_0,0;x)$ should achieve similar approximation errors for the same polynomial degree $N$. Indeed, this is what we find numerically: An example comparison is shown in Fig. \ref{fig:poly_sign_compare}. For this test, we required that $g^-(k,0;x)$ and $f^-(k_0,0;x)$ approximate $\mathrm{sign}(x)$ to within maximum error $\delta=10^{-3}$ on the region $D_\kappa(0)$ with $\kappa=0.05$. The resulting value of the scale parameter is $k\approx k_0\approx 94$. 

From the numerically computed error on the Chebyshev interpolants, we derive a lower bound on the error of the optimal polynomial approximation to $g^-(k,0;x)$ using Theorem \ref{thm:lebesgue constants}. All three polynomial approximations that we consider yield comparable errors. (The black squares are barely visible in Fig. \ref{fig:poly_sign_compare} as they are hidden beneath the corresponding blue circles and red triangles.) Hence we conclude that both $f^-_N(k,0;x)$ and $g^-_N(k,0;x)$ are nearly optimal polynomial approximations of $\mathrm{sign}(x)$.

Moreover, we see that the theoretical upper bound on the error in Eq. \eqref{eq:poly_error} is rather tight, and becomes tighter with increasing $N$. Thus even for very large $N$, where numerical estimates of the error may be computationally difficult, one can rely on the theoretical upper bound to provide a tight estimate of the polynomial order $N$ required to achieve the desired approximation error.

For these choices of $a=0$ and $k\approx 94$, $\tilde{\epsilon}=\|g^-(k,0;x)-\tilde{g}^-(k,0;x)\|$ is suppressed by $\exp(-2209\pi^2)\sim 10^{-9469}$. Thus, in a practical application such as this, the difference between the function described by $g^-_N(k,a;x)$ and $g^-(k,a;x)$ is completely negligible.

\subsection{Eigenvalue Thresholding}
The eigenvalue threshold function
\begin{equation}
    \Theta^\mathrm{thr}(a;x)\equiv 1-\Theta^+(a;x)=\left\{\begin{array}{cc}
       1,  & |x|<a, \\
       0,  &  |x|>a,
    \end{array}\right.
\end{equation}
can be used to filter the spectrum of a matrix according to whether the singular values (equivalently, eigenvalues) are above or below the threshold value $a$. We begin by
approximating the discontinuous function $\Theta^\mathrm{thr}(a;x)$ by the smooth functions
\begin{equation}
    \begin{split}
        f^\mathrm{thr}(k_0,a;x)&\equiv 1-f^+(k_0,a;x),\\
        g^\mathrm{thr}(k,a;x)&\equiv 1-g^+(k,a;x).
    \end{split}
\end{equation}
For $0<\kappa<1$ and $0<a<1-\kappa/2$, the values of $k$ and $k_0$ required to approximate $\Theta^\mathrm{thr}(a;x)$ to within error $\delta$ on $D_{\kappa}(a)$ are determined by Eqs. \eqref{eq:k_req_gplus} and \eqref{eq:k0_scale}, respectively. The smooth functions $f^\mathrm{thr}(k_0,a;x)$ and $g^\mathrm{thr}(k,a;x)$ can then be approximated by the polynomials
\begin{equation}
    \begin{split}
        f_N^\mathrm{thr}(k_0,a;x)&\equiv 1-f_N^+(k_0,a;x),\\
        g_N^\mathrm{thr}(k,a;x)&\equiv 1-g_N^+(k,a;x).
    \end{split}
\end{equation}

\begin{figure}
    \centering
    \includegraphics[scale=0.42]{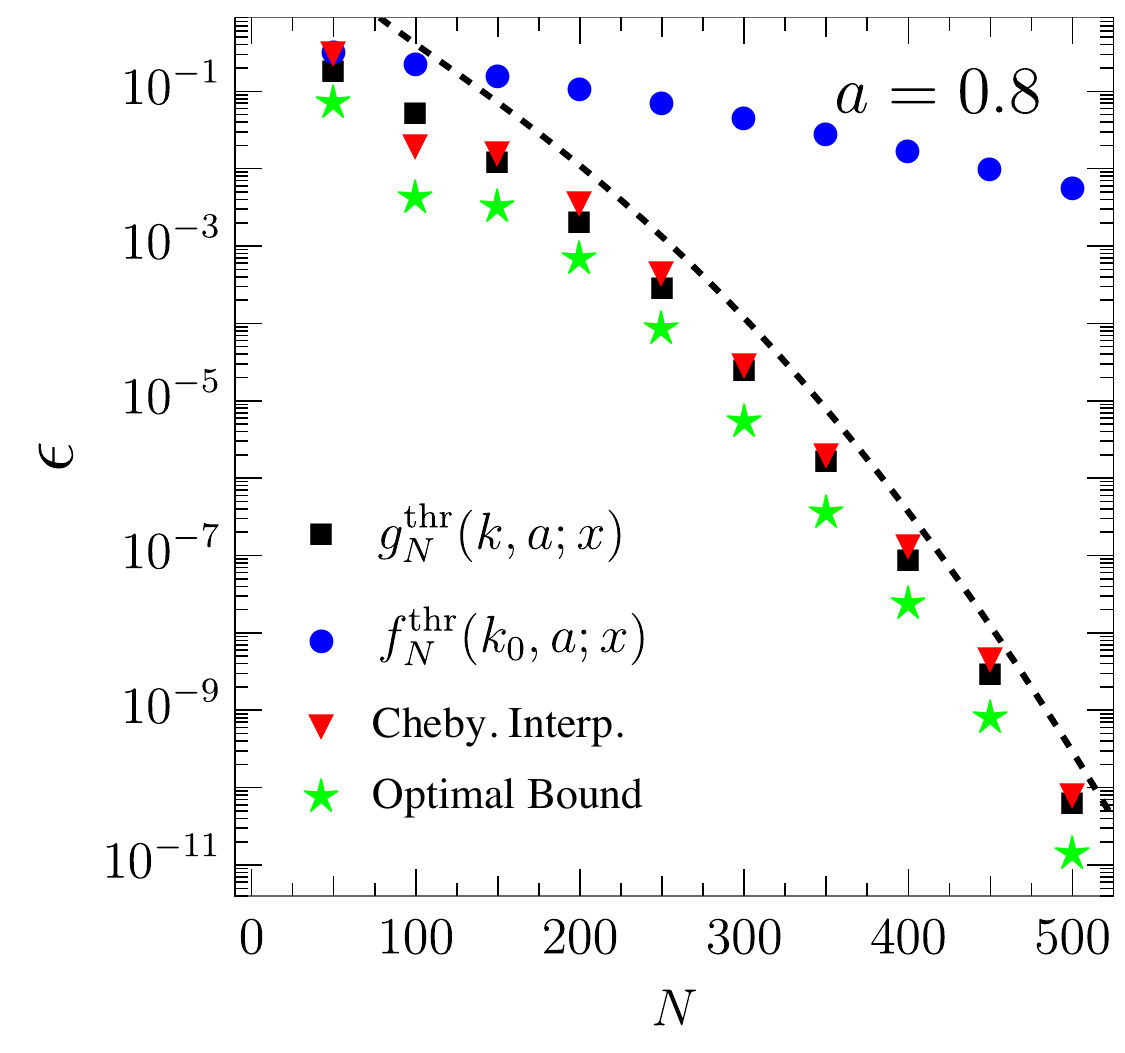}
    \caption{As in Fig. \ref{fig:poly_sign_compare} but for different polynomial approximations of the eigenvalue threshold function $\Theta^\mathrm{thr}(x)$. Black squares correspond to the polynomial $g_N^\mathrm{thr}(k;x)$; blue circles correspond to the polynomial $f_N^\mathrm{thr}(k_0;x)$; red triangles correspond to the degree-$N$ Chebyshev interpolant of $g^\mathrm{thr}(k;x)$; green stars represent a theoretical lower bound on the error of the optimal polynomial approximation to $g^\mathrm{thr}(k;x)$. The dashed black line denotes the theoretical upper bound on the error of $g_N^\mathrm{thr}(k;x)$ from Eq. \eqref{eq:poly_error}. The value of the scaling parameters $k$, $k_0$ were determined by requiring $\delta\leq 10^{-3}$ with $\kappa=0.05$ (see text).}
    \label{fig:poly_thr_compare}
\end{figure}

Figure \ref{fig:poly_thr_compare} compares the errors in the polynomial approximations of the eigenvalue threshold function for an example calculation with threshold $a=0.8$. As in the previous section, we adopt $\kappa=0.05$ and require the error $\delta\leq10^{-3}$, yielding scaling parameter values of $k\approx 58$, $k_0\approx 94$. Based on the discussion in Sec. \ref{sec:remarks}, we expect that for $f_M^\mathrm{thr}(k_0,a;x)$ to achieve the same polynomial approximation error as $g_N^\mathrm{thr}(k,a;x)$ requires $M/N\approx (1+a)/\sqrt{1-a^2}= 3$ when $a=0.8$. We see this behavior reflected in Fig. \ref{fig:poly_thr_compare}, where the polynomial $g_N^\mathrm{thr}(k,a;x)$ achieves an exponentially smaller error than $f_N^\mathrm{thr}(k_0,a;x)$ at the same order $N$. As in the previous example, the error in the polynomial approximation $g_N^\mathrm{thr}(k,a;x)$ is competitive with the error in the corresponding Chebyshev interpolant. Therefore, the numerics support our assertion that $g_N^\mathrm{thr}(k,a;x)$ is a nearly optimal polynomial approximation of the eigenvalue threshold function.

The theoretical bound on the error $\epsilon$ from Eq. \eqref{eq:poly_error} is again relatively tight, though a bit looser than in Fig. \ref{fig:poly_sign_compare}. This is to be expected, as Eq. \eqref{eq:poly_error} was derived by assuming that the bound $|\sqrt{1-a^2}U_{2n}(a)|\leq 1$ is saturated. This is true when $a=0$ but not true for generic $a>0$.

\subsection{Phase Estimation}
\label{sec:QPE}
Within the QSVT framework, quantum phase estimation can be performed by applying the transformation \cite{PRXQuantum.2.040203}
\begin{equation}
    \Theta^\mathrm{QPE}(x)\equiv1-2\Theta^\pm(1/\sqrt{2};x),
\end{equation}
which can be approximated by the smooth functions
\begin{equation}
\begin{split}
    f^\mathrm{QPE}(k;x)&\equiv 1-2f^+(k,1/\sqrt{2};x),\\
    g^\mathrm{QPE}(k;x)&\equiv1-2g^+(k,1/\sqrt{2};x).
    \label{eq:f_g_qpe}
\end{split}
\end{equation}
Thus the required value of the shift parameter is $a=1/\sqrt{2}$. We define the corresponding degree-$N$ even polynomial approximations
\begin{equation}
\begin{split}
    f_N^\mathrm{QPE}(k_0;x)&\equiv  1-2f^+_N(k_0,1/\sqrt{2};x),\\
    g_N^\mathrm{QPE}(k;x)&\equiv 1-2g_N^+(k;1/\sqrt{2};x).
    \label{eq:f_g_N_QPE}
\end{split}
\end{equation}
In Eqs. \eqref{eq:f_g_qpe} and \eqref{eq:f_g_N_QPE}, the factor of two multiplying $g^+(k,a;x)$ and $g_N^+(k;a;x)$ results in an additional factor of two in the error estimates of Theorems \ref{thm:g_plus_fit_proof}, \ref{thm:gtilde_even}, and \ref{thm:gpoly_even}. That is, to achieve errors $\|g^\mathrm{QPE}(k;x)-\Theta^\mathrm{QPE}(x)\|_{D_{\kappa(a)}}\leq \delta$ and $\|g_N^\mathrm{QPE}(k;x)-g^\mathrm{QPE}(k;x)\|\leq\epsilon$, one should adopt values of $k$ and $N$ corresponding to $\delta/2$ and $\epsilon/2$ in Eqs. \eqref{eq:k_req_gminus} and \eqref{eq:N_eps_plus}, respectively. Analogous statements hold for $f^\mathrm{QPE}(k_0;x)$ and $f^\mathrm{QPE}_N(k_0;x)$.

\begin{figure}
    \centering
    \includegraphics[scale=0.42]{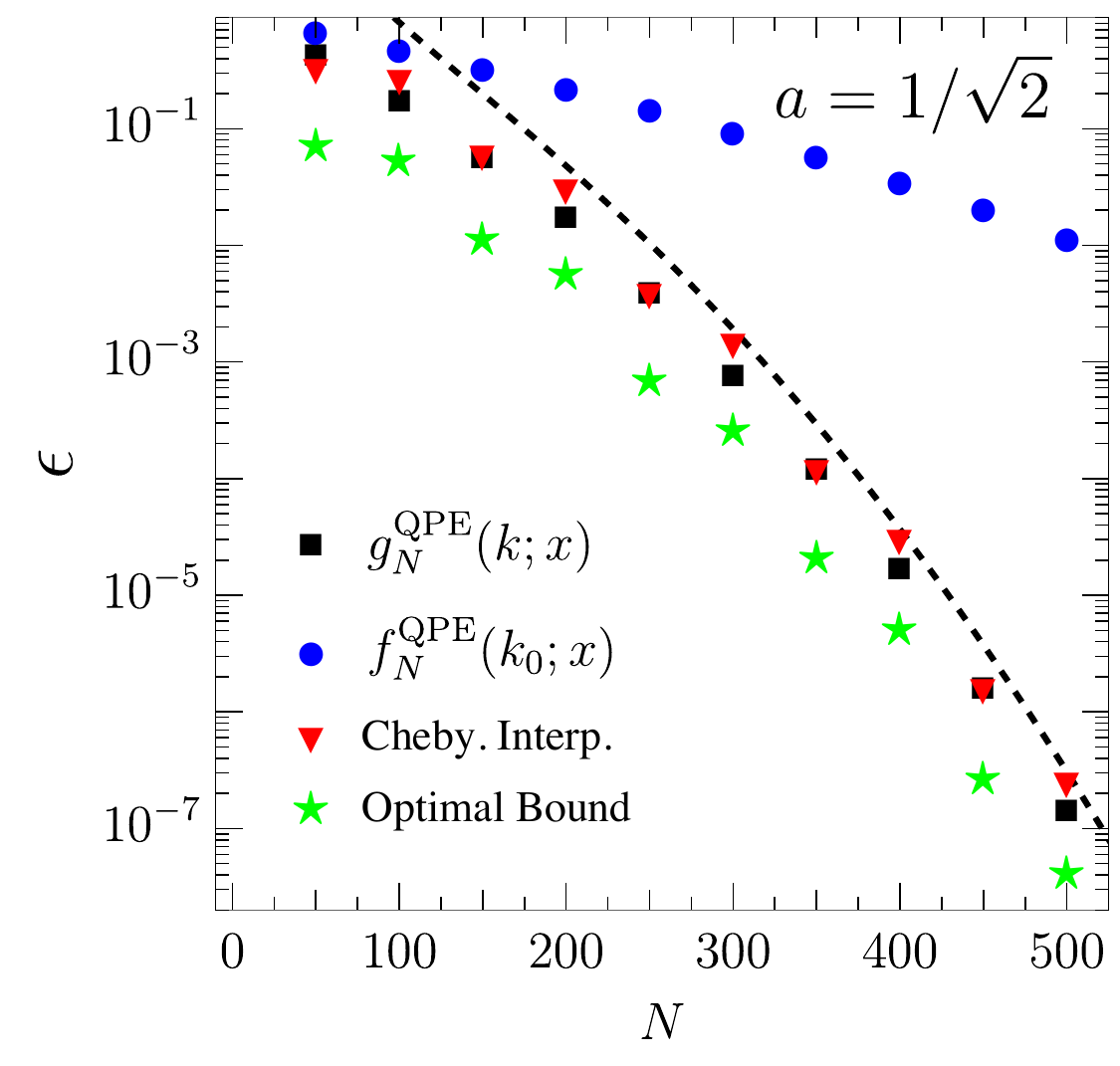}
    \caption{As in Fig. \ref{fig:poly_sign_compare} but for different polynomial approximations of the QPE function $\Theta^\mathrm{QPE}(x)$. Black squares correspond to the polynomial $g_N^\mathrm{QPE}(k;x)$; blue circles correspond to the polynomial $f_N^\mathrm{QPE}(k_0;x)$; red triangles correspond to the degree-$N$ Chebyshev interpolant of $g^\mathrm{QPE}(k;x)$; green stars represent a theoretical lower bound on the error of the optimal polynomial approximation to $g^\mathrm{QPE}(k;x)$. The dashed black line denotes the theoretical upper bound on the error of $g_N^\mathrm{QPE}(k;x)$ from Eq. \eqref{eq:poly_error}. The value of the scaling parameters $k$, $k_0$ were determined by requiring $\delta\leq 10^{-3}$ with $\kappa=0.05$ (see text).}
    \label{fig:poly_qpe_compare}
\end{figure}

Figure \ref{fig:poly_qpe_compare} shows the resulting error in the polynomial approximations $g_N^\mathrm{QPE}(k;x)$ and $f_N^\mathrm{QPE}(k_0;x)$ as a function of the polynomial order $N$. Again, we fix the values of the scaling parameters $k$ and $k_0$ by requiring that $\delta=10^{-3}$ with $\kappa=0.05$. The resulting scaling parameter for $g^\mathrm{QPE}(k;x)$ is $k\approx 72$ whereas to achieve the same error for $f^\mathrm{QPE}(k_0;x)$ requires $k_0\approx 99$. After rescaling, $k'=(1+a)k_0\approx 169$ is the value of the scaling parameter that enters into the polynomial in Eq. \eqref{eq:their_erf_poly}. We expect that for $f_M^\mathrm{QPE}(k;x)$ to achieve the same error $\epsilon$ as $g_N^\mathrm{QPE}(k;x)$ will require $M/N\approx (1+a)/\sqrt{1-a^2}=1+\sqrt{2}\approx 2.4$.

For this choice of $a=1/\sqrt{2}$ and $k\approx 72$, $\tilde{\epsilon}=\|g^+(k,a;x)-\tilde{g}^+(k,a;x)\|$ is suppressed by $\exp(-324\pi^2)\sim 10^{-1389}$. Again, for the reasonable choices of $\kappa=0.05$, $\delta=10^{-3}$, the resulting value of scaling parameter $k$ is such that the difference between $g^+(k,a;x)$ and $\tilde{g}^+(k,a;x)$ is completely negligible. Thus, the polynomial $g_N^\mathrm{QPE}(k;x)$ is essentially equal to the Chebyshev projection of $g^\mathrm{QPE}(k;x)$ and should therefore be nearly optimal.

In Fig. \ref{fig:poly_qpe_compare}, the polynomial $g_N^\mathrm{QPE}(k;x)$ shows comparable errors to the Chebyshev interpolant, supporting our conclusion that $g_N^\mathrm{QPE}(k;x)$ is a nearly optimal polynomial approximation of the QPE function. Again, the theoretical error bound appears tight, especially with increasing $N$. By comparison, the polynomial $f_N^\mathrm{QPE}(k_0;x)$ is far from optimal, yielding errors that are exponentially worse than $g_N^\mathrm{QPE}(k;x)$ at the same order.

\subsection{Linear amplitude amplification}
The aim of linear amplitude amplification \cite{low_2017,PRXQuantum.2.040203} is to multiply input values in the range $[-\Gamma,\Gamma]$ by the constant value $1/2\Gamma$ for some $\Gamma\in(0,1/2]$. The discontinuous function
\begin{equation}
    \Theta^\mathrm{LA}(\Gamma;x)=\frac{x}{2\Gamma}\left[1-\Theta^+(2\Gamma;x)\right],
\end{equation}
performs the desired rescaling for $x\in[-2\Gamma,2\Gamma]$. Values outside of this range are set to zero. We can smoothly approximate $\Theta^\mathrm{LA}(\Gamma;x)$ by 
\begin{equation}
\begin{split}
        f^\mathrm{LA}(k,\Gamma;x)&\equiv\frac{x}{2\Gamma}\left[1-f^+(k,2\Gamma;x)\right],\\
        g^\mathrm{LA}(k,\Gamma;x)&\equiv\frac{x}{2\Gamma}\left[1-g^+(k,2\Gamma;x)\right].
\end{split}
\end{equation}
The function $\Theta^\mathrm{LA}(\Gamma;x)$ rescales values within the range $x\in[-2\Gamma,2\Gamma]$, but our original goal was only to rescale values within $[-\Gamma,\Gamma]$. Thus, it suffices to choose $\kappa=\Gamma$, so that for the appropriate value of scaling parameter $k$, our smooth approximation of $\Theta^\mathrm{LA}(\Gamma;x)$ will be accurate within error $\delta$ for $x\in[-\Gamma,\Gamma]$. Note that we are not particularly concerned with the behavior of $g^\mathrm{LA}(k,\Gamma;x)$ for $|x|>\Gamma$; the only crucial requirement is that of boundedness, $|g^\mathrm{LA}(k,\Gamma;x)|\leq 1$ for all $x\in[-1,1]$, which follows from the fact that $|g^+(k,2\Gamma;x)|\leq 1$ and $|x/2\Gamma|\leq 1$ for $x\in[-1,1]$. 

\begin{equation}
\begin{split}
        f^\mathrm{LA}_N(k,\Gamma;x)&\equiv\frac{x}{2\Gamma}\left[1-f^+_N(k,2\Gamma;x)\right]\\
        g^\mathrm{LA}_N(k,\Gamma;x)&\equiv\frac{x}{2\Gamma}\left[1-g^+_N(k,2\Gamma;x)\right]
\end{split}
\end{equation}

As we have multiplied the functions $g^+(k,2\Gamma;x)$ by the linear function $x/2\Gamma$, the error bounds derived in Theorems \ref{thm:g_plus_fit_proof} and \ref{thm:gpoly_even} do not apply. Fortunately, through a slight modification of the proof of Theorem \ref{thm:g_plus_fit_proof}, we find that for all $\delta>0$, choosing
\begin{equation}
    k=\frac{1}{\sqrt{2}\eta}\sqrt{W\left(\frac{1}{2\pi \delta^2}\right)},
\end{equation}
guarantees that 
\begin{equation}
    \max_{|x|<\Gamma}\left|\Theta^\mathrm{LA}(\Gamma;x)-g^\mathrm{LA}(k,\Gamma;x)\right|\leq \delta.
    \end{equation}

Again, we are not concerned with the deviation between $\Theta^\mathrm{LA}(k;x)$ and $g^\mathrm{LA}(k;x)$ for $|x|>\Gamma$, leading to the one-sided bound above. Similarly, one can show that
\begin{equation}
    \|g_N^\mathrm{LA}(k,\Gamma;x)-g^\mathrm{LA}(k,\Gamma;x)\|\leq \frac{4}{\pi}\frac{1}{2\Gamma}\frac{k^2}{N^2}e^{-N^2/4k^2},
\end{equation}
so that for all $\epsilon>0$, choosing 
\begin{equation}
    N=2\ceil*{k\sqrt{W\left(1/(2\pi\Gamma \epsilon)\right)}},
\end{equation}
guarantees that $\|g_N^\mathrm{LA}(k,\Gamma;x)-g^\mathrm{LA}(k,\Gamma;x)\|<\epsilon$. 

With these results established, we now proceed with a numerical example, choosing $\Gamma=\kappa=1/4$. Setting the error tolerance to $\delta=10^{-3}$ yields $k\approx 8$, $k_0\approx 9$. Compared to previous examples, the required scaling parameters values are modest due to the relatively large window parameter $\kappa=0.25$.

The resulting approximation errors $\epsilon$ are shown in Fig. \ref{fig:poly_linamp_compare}. Due to the smaller scaling parameters $k$, $k_0$, all three polynomial approximations converge much faster than in previous examples (note the scale of the $x$ axis). As before, we find that $g^\mathrm{LA}_N(k,\Gamma;x)$ yields comparable errors to the Chebyshev interpolant, providing numerical evidence that $g^\mathrm{LA}_N(k,\Gamma;x)$ is a nearly optimal polynomial approximation of the linear amplitude amplification function $\Theta^\mathrm{LA}(\Gamma;x)$. As the linear amplification function corresponds to a shifted sign function (multiplied by a linear function), we expect that the corresponding polynomial $f_N^\mathrm{LA}(k_0,\Gamma;x)$ will be far from optimal; indeed, Fig. \ref{fig:poly_thr_compare} supports this conclusion.

\begin{figure}
    \centering
    \includegraphics[scale=0.42]{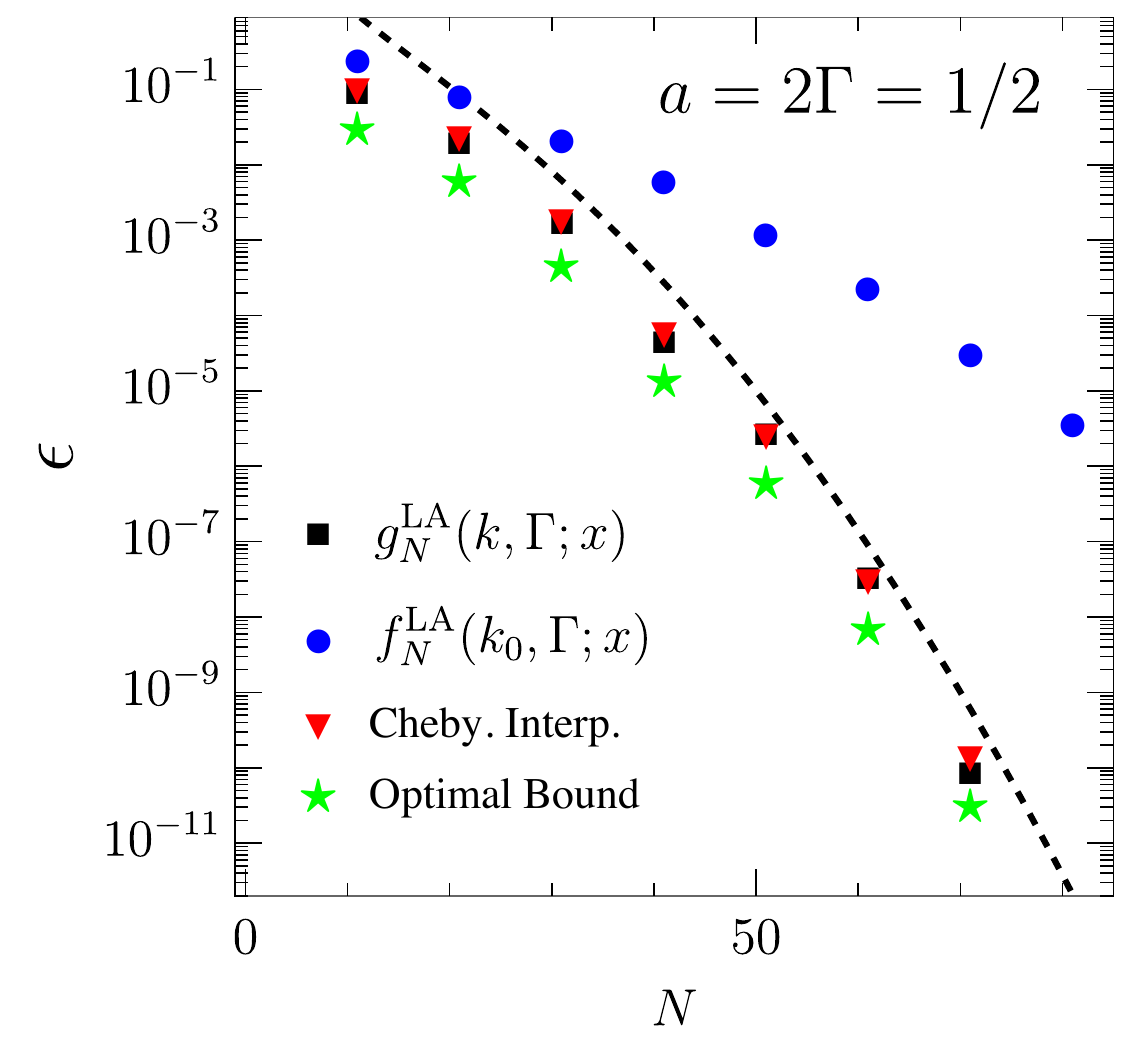}
        \caption{As in Fig. \ref{fig:poly_sign_compare} but for different polynomial approximations of the linear amplification function $\Theta^\mathrm{LA}(\Gamma;x)$. Black squares correspond to the polynomial $g_N^\mathrm{LA}(k,\Gamma;x)$; blue circles correspond to the polynomial $f_N^\mathrm{LA}(k_0,\Gamma;x)$; red triangles correspond to the degree-$N$ Chebyshev interpolant of $g^\mathrm{LA}(k,\Gamma;x)$; green stars represent a theoretical lower bound on the error of the optimal polynomial approximation to $g^\mathrm{LA}(k,\Gamma;x)$. The dashed black line denotes the theoretical upper bound on the error of $g_N^\mathrm{LA}(k,\Gamma;x)$ from Eq. \eqref{eq:poly_error}. The value of the scaling parameters $k$, $k_0$ were determined by requiring $\delta\leq 10^{-3}$ with $\kappa=\Gamma=0.25$ (see text).}
    \label{fig:poly_linamp_compare}
\end{figure}

\subsection{Boundedness}
When constructing polynomials $p_N(x)$ in the above examples, we have largely ignored the requirement that QSVT polynomials must be bounded, $|p_N(x)|\leq 1$ for $x\in[-1,1]$. We have noted that each of the relevant smooth function approximations --- $g^\pm(k,a;x)$, $g^\mathrm{QPE}(k;x)$, etc. --- is bounded by $\pm 1$ on this interval. Thus, the corresponding polynomial approximations can only violate this bound to the extent that they deviate from the smooth function, which is precisely the error $\epsilon=\|p(x)-p_N(x)\|$. Thus we can ensure that each polynomial is properly bounded by a simple rescaling $p_N(x)\rightarrow p_N(x)/(1+\epsilon)$.

\section{Summary}
In this work, we have provided simple polynomial approximations $g_N^\pm(k,a;x)$ to the even and odd step functions $\Theta^\pm(a;x)$ on the interval $[-1,1]$. As demonstrated above, previously known polynomial approximations of the even and odd step functions $\Theta^\pm(a;x)$ on $[-1,1]$ perform poorly as the location of the step-function discontinuity is shifted away from the origin and toward the endpoints of the interval; that is, as $a\rightarrow 1$. In contrast, we introduced polynomials $g_N^\pm(k,a;x)$ that are nearly optimal approximations of $\Theta^\pm(a;x)$  for all values of $0\leq a<1$. 

In addition to yielding nearly optimal errors, the polynomials $g_N^\pm(k,a;x)$ are described by exceedingly simple Chebyshev coefficients. Moreover, we derived a theoretical upper bound on the approximation error that, based on our numerical investigations, is tight. Thus, the results of Theorems \ref{thm:g_minus_fit_proof} -- \ref{thm:gpoly_even} can be directly utilized to derive tight estimates for the value of the scaling parameter $k$ and the degree $N$ of the polynomial approximation required to achieve total error $\delta+\epsilon$.  To this end, the polynomial approximations $g_N^\pm(k,a;x)$ --- and those derived from them in Sec. \ref{sec:applications} for specific algorithmic applications --- should be useful in practical QSVT calculations on quantum hardware.  

\section*{Acknowledgments}
This work was carried out under the auspices of the National Nuclear Security Administration of the U.S. Department of Energy at Los Alamos National Laboratory under Contract No. 89233218CNA000001. The research presented in this article was partially supported by the Laboratory Directed Research and Development program of Los Alamos National Laboratory under project numbers 20251163PRD3 and 20260043DR.

\bibliographystyle{unsrt}
\bibliography{poly_refs}

\appendix
\begin{widetext}
\section{Proof of Theorem \ref{thm:g_minus_fit_proof}}
\label{app:g_minus_fit_proof}
\begin{proof}
Letting $D_{\kappa(a)}$ be the region defined in Eq. \eqref{eq:D_kappa}, our goal is to bound the error 
    \begin{equation}
        \|\Theta^-(a;x)-g^-(k,a;x)\|_{D_\kappa(a)}=\max(M_1,M_2),
    \end{equation}
    where
    \begin{equation}
        \begin{split}
            M_1&\equiv \max_{|x|\in[0,a-\kappa/2]}|g^-(k,a;x)|,\\
            M_2&\equiv\max_{x\in[a+\kappa/2,1]}|1-g^-(k,a;x)|.
        \end{split}
    \end{equation}
    As both $\Theta^-(a;x)$ and $g^-(k,a;x)$ are odd functions, it suffices to consider $0\leq x\leq 1$. We begin by bounding $M_1$. When $0\leq x\leq a-\kappa/2$, then $\arcsin x < \arcsin a$, so that
    \begin{equation}
        \begin{split}
                |g^-(k,a;x)|&=\frac{1}{\sqrt{\pi}}\int_{k(\arcsin a-\arcsin x)}^{k(\arcsin a+\arcsin x)}dt~e^{-t^2}\\
                &\leq \frac{1}{2\sqrt{\pi}}\frac{1}{k(\arcsin a - \arcsin x)}e^{-k^2(\arcsin a - \arcsin x)^2},
        \end{split}
    \end{equation}
    which is maximized for $x\in[0,a-\kappa/2]$ when $x=a-\kappa/2$. Therefore
    \begin{equation}
        M_1\leq \frac{1}{2\sqrt{\pi}}\frac{1}{k\eta}e^{-k^2\eta^2},
    \end{equation}
    where $\eta\equiv \arcsin(a) - \arcsin(a-\kappa/2)$. Thus for any $\delta >0$, choosing
    \begin{equation}
        k=\frac{1}{\sqrt{2}\eta}\sqrt{W\left(\frac{1}{2\pi\delta^2}\right)},
    \end{equation}
    guarantees that $M_1\leq \delta$. Next we derive a bound on $M_2$. When $a+\kappa/2\leq x\leq 1$, then $\arcsin x>\arcsin a$ and
    \begin{equation}
    \begin{split}
        |1-g^-(k,a;x)|&=\left|\frac{1}{2}\left\{\mathrm{erfc}[k(\arcsin x -\arcsin a)]+\mathrm{erfc}[k(\arcsin x + \arcsin a)]\right\}\right|\\
        &=\frac{1}{\sqrt{\pi}}\int_{k(\arcsin x - \arcsin a)}^{\infty} dt~e^{-t^2}+\frac{1}{\sqrt{\pi}}\int_{k(\arcsin x + \arcsin a)}^{\infty}dt~ e^{-t^2}\\
        &\leq\frac{2}{\sqrt{\pi}}\int_{k(\arcsin x - \arcsin a)}^{\infty} dt\frac{t}{k(\arcsin x - \arcsin a)}e^{-t^2}\\
        & =\frac{1}{\sqrt{\pi}}\frac{e^{-k^2(\arcsin x-\arcsin a)^2}}{k(\arcsin x - \arcsin a)},
    \end{split}
    \end{equation}
    which is maximized for $x\in[a+\kappa/2,1]$ when $x=a+\kappa/2$. Therefore
    \begin{equation}
        M_2\leq \frac{1}{\sqrt{\pi}}\frac{1}{k\nu}e^{-k^2\nu^2},
    \end{equation}
    where $\nu\equiv \arcsin(a+\kappa/2)-\arcsin(a)$. Thus for any $\delta>0$, choosing
    \begin{equation}
        k=\frac{1}{\sqrt{2}\nu}\sqrt{W\left(\frac{2}{\pi\delta^2}\right)}
    \end{equation}
    guarantees that $M_2\leq \delta$. Putting these results together, choosing 
     \begin{equation}
        k=\mathrm{max}\left[\frac{1}{\sqrt{2}\eta}\sqrt{W\left(\frac{1}{2\pi\delta^2}\right)},\frac{1}{\sqrt{2}\nu}\sqrt{W\left(\frac{2}{\pi\delta^2}\right)}\right]
    \end{equation}
    guarantees that 
    \begin{equation}
        \|\Theta^-(a;x)-g^-(k,a;x)\|_{D_\kappa(a)} \leq \delta.
    \end{equation}

\end{proof}

\section{Proof of Theorem \ref{thm:gtilde_odd}}
\label{app:gtilde_odd}
\begin{proof}
Let $x=\cos\xi$, $a=\cos\alpha$. Then the limit of the polynomial series in Eq. \eqref{eq:gpoly_minus} can be written as
\begin{equation}
    \begin{split}
        \tilde{g}^-(k,a;x)&=\lim_{N\rightarrow \infty}g^-_N(k,a;x)\\
        &=\frac{2}{\pi}\sum_{n=0}^\infty\frac{\sin[(2n+1)(\alpha+\xi)]+\sin[(2n+1)(\alpha-\xi)]}{2n+1}e^{-(2n+1)^2/4k^2}\\
        &=\frac{1}{\pi}\int_0^{\alpha+\xi}dz~\theta_2(z,e^{-1/k^2})+\frac{1}{\pi}\int_0^{\alpha-\xi}dz~\theta_2(z,e^{-1/k^2}),
    \end{split}
\end{equation}
where
\begin{equation}
    \theta_2(z,q)\equiv 2q^{1/4}\sum_{n=0}^{\infty} q^{n(n+1)}\cos[(2n+1)z],
\end{equation}
is the second elliptic theta function. For $0\leq a<1$ and $-1\leq x\leq 1$, it follows that $0<\alpha+\xi\leq3\pi/2$ and $-\pi<\alpha-\xi\leq \pi/2$. Using the transformation properties of the theta functions under modular transformations, we can equivalently write

\begin{equation}
    \theta_2(z,e^{-1/k^2})=k\sqrt{\pi}e^{-k^2z^2}\left[1+\sum_{n=1}^{\infty}(-1)^ne^{-k^2\pi^2n^2}\cosh(2m\pi k^2z)\right],
\end{equation}
which can be integrated term-by-term to yield
\begin{equation}
\begin{split}
        &\frac{1}{\pi}\int_0^{\alpha\pm\xi}dz~\theta_2(z,e^{-1/k^2})=\frac{1}{2}\mathrm{erf}[k(\alpha\pm\xi)]+\frac{1}{2}\sum_{n=1}^{\infty}(-1)^n\left\{\mathrm{erf}[k(n\pi+\alpha\pm\xi)]-\mathrm{erf}[k(n\pi-\alpha\mp\xi)]\right\}.
\end{split}
\end{equation}
Recalling that 
\begin{equation}
\begin{split}
        g^-(k,a;x)&=\frac{1}{2}\left\{\mathrm{erf}[k(\arcsin x+\arcsin a)]+\mathrm{erf}[k(\arcsin x - \arcsin a]\right\}\\
        &=\frac{1}{2}\left\{\mathrm{erf}[k(\pi-\alpha-\xi)]+\mathrm{erf}[k(\alpha-\xi)]\right\},
        \end{split}
\end{equation}
we can write
\begin{equation}
\begin{split}
    \tilde{g}^-(k,a;x)&=g^-(k,a;x)+R_1(k,a;x)+R_2(k,a;x)+R_3(k,a;x),
\end{split}
\end{equation}
where 
\begin{equation}
\begin{split}
    R_1(k,a;x)&\equiv \frac{1}{2}\left\{\mathrm{erf}[k(\alpha+\xi)]-\mathrm{erf}[k(\pi+\alpha+\xi)]\right\},\\
    R_2(k,a;x)&\equiv \frac{1}{2}\sum_{n=1}^{\infty}(-1)^n\left\{\mathrm{erf}[k(n\pi+\alpha-\xi)]-\mathrm{erf}[k(n\pi-\alpha+\xi)]\right\},\\
    R_3(k,a;x)&\equiv \frac{1}{2}\sum_{n=2}^{\infty}(-1)^n\left\{\mathrm{erf}[k(n\pi+\alpha+\xi)]-\mathrm{erf}[k(n\pi-\alpha-\xi)]\right\}.
\end{split}
\end{equation}
Therefore the error is bounded by
\begin{equation}
    \begin{split}
        \|\tilde{g}^-(k,a;x)-g^-(k,a;x)\|&=\|R_1(k,a;x)+R_2(k,a;x)+R_3(k,a;x)\|\\
        &\leq \|R_1(k,a;x)\|+\|R_2(k,a;x)\|+\|R_3(k,a;x)\|
    \end{split}
\end{equation}

\begin{equation}
\begin{split}
    |R_1(k,a;x)|&=\frac{1}{2}|\mathrm{erf}[k(\alpha+\xi)]-\mathrm{erf}[k(\pi+\alpha+\xi)]|\\
    &= \frac{1}{\sqrt{\pi}}\int_{k(\alpha+\xi)}^{k(\pi+\alpha+\xi)}dt~e^{-t^2}\\
    &\leq \frac{1}{2\sqrt{\pi}k(\alpha+\xi)}e^{-k^2(\alpha+\xi)^2}.
\end{split}
\end{equation}
This error is maximized when $\xi$ attains its minimal value $\xi=0$ corresponding to $x=1$. Therefore
\begin{equation}
    \|R_1(k,a;x)\|\leq \frac{1}{2\sqrt{\pi}k\alpha}e^{-k^2\alpha^2}.
\end{equation}

The second term can be bounded by noting that $R_2(k,a;x)$ is an alternating series where successive terms decrease in magnitude. Therefore, the absolute value of the sum is bounded by the first term
\begin{equation}
    \begin{split}
     |R_2(k,a;x)|&=\left|\frac{1}{2}\sum_{n=1}^{\infty}(-1)^n\left\{\mathrm{erf}[k(n\pi+\alpha-\xi)]-\mathrm{erf}[k(n\pi-\alpha+\xi)\right\}\right|\\
     &\leq \frac{1}{2}|\mathrm{erf}[k(\pi+\alpha-\xi)]-\mathrm{erf}[k(\pi-\alpha+\xi)]|\\
     &=\left\{\begin{array}{cc}
        \frac{1}{\sqrt{\pi}}\int_{k(\pi-\alpha+\xi)}^{k(\pi+\alpha-\xi)}dt~e^{-t^2},  & \alpha>\xi \\
          \frac{1}{\sqrt{\pi}}\int_{k(\pi+\alpha-\xi)}^{k(\pi-\alpha+\xi)}dt~e^{-t^2},  & \alpha<\xi  
     \end{array}\right.\\
     &\leq \left\{\begin{array}{cc}
       \frac{1}{\sqrt{\pi}}\int_{k(\pi-\alpha+\xi)}^{k(\pi+\alpha-\xi)}dt~\frac{t}{k(\pi-\alpha+\xi)}e^{-t^2} & \alpha>\xi  \\
        \frac{1}{\sqrt{\pi}}\int_{k(\pi+\alpha-\xi)}^{k(\pi-\alpha+\xi)}dt~\frac{t}{k(\pi+\alpha-\xi)}e^{-t^2} & \alpha<\xi
     \end{array}\right. \\
      &\leq \left\{\begin{array}{cc}
       \frac{1}{2\sqrt{\pi}k(\pi-\alpha+\xi)}e^{-k^2(\pi-\alpha+\xi)^2} & \alpha>\xi  \\
        \frac{1}{2\sqrt{\pi}k(\pi+\alpha-\xi)}e^{-k^2(\pi+\alpha-\xi)^2} & \alpha<\xi
     \end{array}\right. \\
    \end{split}
\end{equation}
When $\alpha>\xi$, the bound is maximized when $\xi=0$. When $\alpha<\xi$, the bound is maximized when $\xi=\pi$. Therefore
\begin{equation}
    \|R_2(k,a;x)\|\leq \left\{\begin{array}{lc}
       \frac{1}{2\sqrt{\pi}k(\pi-\alpha)}e^{-k^2(\pi-\alpha)^2}, & \alpha>\xi  \\
        \frac{1}{2\sqrt{\pi}k\alpha}e^{-k^2\alpha^2}, & \alpha<\xi
    \end{array}\right.
\end{equation}
As $0\leq a<1$ implies $0<\alpha\leq \pi/2$, it follows that the second bound is always larger, and therefore we take
\begin{equation}
   \|R_2(k,a;x)\|\leq \frac{1}{2\sqrt{\pi}k\alpha}e^{-k^2\alpha^2}.
\end{equation}

For the final term, $0<\alpha+\xi<2\pi$. As in the previous term, $R_3(k,a;x)$ is an alternating series with terms that decrease in magnitude and is therefore bounded by the first term.
\begin{equation}
    \begin{split}
     |R_3(k,a;x)|&\leq \frac{1}{2}\left\{\mathrm{erf}[k(2\pi+\alpha+\xi)]-\mathrm{erf}[k(2\pi-\alpha-\xi)]\right\}\\
     &=
        \frac{1}{\sqrt{\pi}}\int_{k(2\pi-\alpha-\xi)}^{k(2\pi+\alpha+\xi)}dt~e^{-t^2}\\
     &\leq 
       \frac{1}{\sqrt{\pi}}\int_{k(2\pi-\alpha+\xi)}^{k(\pi-\alpha-\xi)}dt~\frac{t}{k(2\pi-\alpha-\xi)}e^{-t^2}  \\
      &\leq 
       \frac{1}{2\sqrt{\pi}k(2\pi-\alpha-\xi)}e^{-k^2(2\pi-\alpha-\xi)^2}  
    \end{split}
\end{equation}
This is maximized when $\xi=\pi$, so that
\begin{equation}
    \|R_3(k,a;x)\|\leq  \frac{1}{2\sqrt{\pi}k(\pi-\alpha)}e^{-k^2(\pi-\alpha)^2}
\end{equation}
Combining these three terms, we arrive at
\begin{equation}
    \begin{split}
        \|\tilde{g}^-(k,a;x)-g^-(k,a;x)\|&\leq \frac{1}{\sqrt{\pi}k\alpha}e^{-k^2\alpha^2}+\frac{1}{2\sqrt{\pi}k(\pi-\alpha)}e^{-k^2(\pi-\alpha)^2}\\
        &\leq \frac{3}{2\sqrt{\pi}k\alpha}e^{-k^2\alpha^2}.
    \end{split}
\end{equation}
Thus for any $\tilde{\epsilon}>0$, choosing 
\begin{equation}
    k=\frac{1}{\sqrt{2}\alpha}\sqrt{W\left(\frac{9}{2\pi\tilde{\epsilon}^2}\right)},
\end{equation}
guarantees that 
\begin{equation}
    \|\tilde{g}^-(k,a;x)-g^-(k,a;x)\|\leq \tilde{\epsilon}.
\end{equation}

\end{proof}

\section{Proof of Theorem \ref{thm:gpoly_odd}}
\label{app:gpoly_odd}
\begin{proof}
Our goal is to bound
\begin{equation}
          \|\tilde{g}^-(k,a;x)-g^-_N(k,a;x)\|=\left\|\frac{4}{\pi}\sqrt{1-a^2}\sum_{n=(N+1)/2}^{\infty}\frac{U_{2n}(a)}{2n+1}e^{-(2n+1)^2/4k^2}T_{2n+1}(x)\right\|\\
\end{equation}
Noting that $|\sqrt{1-a^2}U_{2n}(a)|\leq 1$ and $|T_{2n+1}(x)|\leq 1$, we have
\begin{equation}
\begin{split}
      \|\tilde{g}^-(k,a;x)-g^-_N(k,a;x)\|&\leq \frac{4}{\pi}\sum_{n=(N+1)/2}^{\infty}\frac{1}{2n+1}e^{-(2n+1)^2/4k^2}\\
      &\leq \frac{4}{\pi}\int_{(N-1)/2}^{\infty}dt ~\frac{1}{2t+1}e^{-(2t+1)^2/4k^2}\\
      &=\frac{1}{\pi}\mathrm{E_1}\left[\frac{N^2}{4k^2}\right]\\
      &\leq \frac{4}{\pi}\frac{k^2}{N^2}e^{-N^2/4k^2},
      \end{split}
\end{equation}
where 
\begin{equation}
    E_1(x)\equiv \int_x^{\infty}dt\frac{e^{-t}}{t},
\end{equation}
is the exponential integral function. Thus for any $\epsilon>0$, choosing 
\begin{equation}
    N=2\ceil*{k\sqrt{W(1/(\pi\epsilon))}}+1,
\end{equation}
guarantees that 
\begin{equation}
     \|\tilde{g}^-(k,a;x)-g^-_N(k,a;x)\|\leq\epsilon.
\end{equation}
\end{proof}

\section{Proof of Theorem \ref{thm:g_plus_fit_proof}}
\label{app:g_plus_fit_proof}
Consider
    \begin{equation}
        \|\Theta^+(a;x)-g^+(k,a;x)\|_{D_\kappa(a)}=\max(M_1,M_2),
    \end{equation}
    where
    \begin{equation}
        \begin{split}
            M_1&\equiv \max_{|x|\in[0,a-\kappa/2]}|g^+(k,a;x)|,\\
            M_2&\equiv\max_{x\in[a+\kappa/2,1]}|1-g^+(k,a;x)|.
        \end{split}
    \end{equation}

First consider the region where $|x|\leq a-\kappa/2$ and
    \begin{equation}
    \begin{split}
        |g^+(k,a;x)|&=\left|\frac{1}{2}\left\{\mathrm{erfc}[k(\arcsin a +\arcsin x)]+\mathrm{erfc}[k(\arcsin a - \arcsin x)]\right\}\right|\\
        &=\frac{1}{\sqrt{\pi}}\int_{k(\arcsin a + \arcsin x)}^{\infty} dt~e^{-t^2}+\frac{1}{\sqrt{\pi}}\int_{k(\arcsin a - \arcsin x)}^{\infty}dt~ e^{-t^2}\\
        &\leq\frac{2}{\sqrt{\pi}}\int_{k(\arcsin a - \arcsin x)}^{\infty} dt\frac{t}{k(\arcsin a - \arcsin x)}e^{-t^2}\\
        & =\frac{1}{\sqrt{\pi}}\frac{e^{-k^2(\arcsin a-\arcsin x)^2}}{k(\arcsin a - \arcsin x)},
    \end{split}
    \end{equation}
    which is maximized at $x=a-\kappa/2$. Therefore
    \begin{equation}
        M_1=\frac{1}{\sqrt{\pi}}\frac{1}{k\eta}e^{-k^2\eta^2},
    \end{equation}
    where $\eta\equiv \arcsin(a)-\arcsin(a-\kappa/2)$. Thus for any $\delta>0$, it suffices to choose
    \begin{equation}
        k=\frac{1}{\sqrt{2}\eta}\sqrt{W\left(\frac{2}{\pi\delta^2}\right)},
    \end{equation}
    to guarantee that $M_1\le \delta$. Next consider the region where $a+\kappa/2<|x|\leq 1$ and
    \begin{equation}
        \begin{split}
                |1-g^+(k,a;x)|&=\left|\frac{1}{2}\left\{\mathrm{erf}[k(\arcsin x+\arcsin a)]-\mathrm{erf}[k(\arcsin x-\arcsin a)]\right\}\right|\\
                &=\frac{1}{\sqrt{\pi}}\int_{k(\arcsin x-\arcsin a)}^{k(\arcsin x+\arcsin a)}dt~e^{-t^2}\\
                &\leq\frac{1}{\sqrt{\pi}}\int_{k(\arcsin x-\arcsin a)}^{k(\arcsin x+\arcsin a)}dt\frac{t}{k(\arcsin x-\arcsin a)}e^{-t^2}\\
                &=\frac{1}{2\sqrt{\pi}}\frac{1}{k(\arcsin x - \arcsin a)}\left[e^{-k^2(\arcsin x - \arcsin a)^2}-e^{-k^2(\arcsin x + \arcsin a)^2}\right]\\
                &\leq \frac{1}{2\sqrt{\pi}}\frac{1}{k(\arcsin x - \arcsin a)}e^{-k^2(\arcsin x - \arcsin a)^2},
        \end{split}
    \end{equation}
    which is maximized when $x=a+\kappa/2$. Therefore
    \begin{equation}
        M_2\leq \frac{1}{2\sqrt{\pi}}\frac{1}{k\nu}e^{-k^2\nu^2},
    \end{equation}
    where $\nu\equiv \arcsin(a+\kappa/2) - \arcsin(a)$. Thus for any $\delta>0$, choosing 
    \begin{equation}
        k=\frac{1}{\sqrt{2}\nu}\sqrt{W\left(\frac{1}{2\pi\delta^2}\right)},
    \end{equation}
    guarantees that $M_2\leq \delta$.  Putting these results together, choosing 
     \begin{equation}
        k=\mathrm{max}\left[\frac{1}{\sqrt{2}\eta}\sqrt{W\left(\frac{2}{\pi\delta^2}\right)},\frac{1}{\sqrt{2}\nu}\sqrt{W\left(\frac{1}{2\pi\delta^2}\right)}\right]
    \end{equation}
    guarantees that 
    \begin{equation}
        \|\Theta^+(a;x)-g^+(k,a;x)\|_{D_\kappa(a)} \leq \delta.
    \end{equation}

\section{Proof of Theorem \ref{thm:gtilde_even}}
\label{app:gtilde_even}
\begin{proof}
Letting $x=\cos\xi$ and $a=\cos\alpha$, we begin by writing
\begin{equation}    
    \begin{split}
        \tilde{g}^+(k,a;x)&=\frac{2\alpha}{\pi}+\frac{4}{\pi}\sum_{n=1}^{\infty}\frac{\sin(2n\alpha)}{2n}e^{-n^2/k^2}\cos(2n\xi)\\
        &= \frac{1}{\pi}\int_0^{\alpha+\xi}dz~\theta_3(z,e^{-1/k^2})+\frac{1}{\pi}\int_0^{\alpha-\xi}dz~\theta_3(z,e^{-1/k^2}),
    \end{split}
\end{equation}
where \begin{equation}
    \theta_3(z,q)\equiv 1+2\sum_{n=1}^{\infty}q^{n^2}\cos(2nz),
\end{equation}
is the third elliptic theta function. As before, via modular transformation we can write
\begin{equation}
    \theta_3(z,e^{-1/k^2})=k\sqrt{\pi}e^{-k^2z^2}\left[1+2\sum_{n=1}^\infty e^{-n^2\pi^2k^2}\cosh(2n\pi k^2z)\right],
\end{equation}
which can be integrated term-by-term to yield
\begin{equation}
\begin{split}
   \tilde{g}^+(k,a;x) &= \frac{1}{2}\left\{\mathrm{erf}[k(\alpha+\xi)]+\mathrm{erf}[k(\alpha-\xi)]\right\}+\frac{1}{2}\sum_{n=1}^{\infty}\left\{\mathrm{erf}[k(n\pi+\alpha-\xi)]-\mathrm{erf}[k(n\pi-\alpha+\xi)]\right\}\\
   &+\frac{1}{2}\sum_{n=1}^{\infty}\left\{\mathrm{erf}[k(n\pi+\alpha+\xi)]-\mathrm{erf}[k(n\pi-\alpha-\xi)]\right\}\\
   &=g^+(k,a;x)+R_1+R_2+R_3+R_4,
   \end{split}
\end{equation}
where 
\begin{equation}
    \begin{split}
        R_1&\equiv -\frac{1}{2}\left\{\mathrm{erfc}[k(\alpha+\xi)]+\mathrm{erfc}[k(\pi+\alpha+\xi)]\right\},\\
        R_2&\equiv \frac{1}{2}\left\{\mathrm{erf}[k(\pi+\alpha-\xi)-\mathrm{erf}[k(\pi-\alpha+\xi)]\right\},\\
        R_3&\equiv \frac{1}{2}\sum_{n=2}^{\infty}\left\{\mathrm{erf}[k(n\pi+\alpha-\xi)-\mathrm{erf}[k(n\pi-\alpha+\xi)]\right\},\\
        R_4&\equiv\frac{1}{2}\sum_{n=2}^{\infty}\left\{\mathrm{erf}[k(n\pi+\alpha+\xi)-\mathrm{erf}[k(n\pi-\alpha-\xi)]\right\},
    \end{split}
\end{equation}

\begin{equation}
    \begin{split}
        |R_1|&\leq \mathrm{erfc}[k(\alpha+\xi)]\\
        &\leq \frac{1}{\sqrt{\pi} k \alpha}e^{-k^2\alpha^2},
    \end{split}
\end{equation}
with solution
\begin{equation}
    k_1=\frac{1}{\sqrt{2}\alpha}\sqrt{W\left(\frac{2}{\pi\epsilon^2}\right)}
\end{equation}

\begin{equation}
\begin{split}
    |R_2|&\leq \left\{\begin{array}{cc}
      \frac{1}{\sqrt{\pi}}\int_{k(\pi+\alpha-\xi)}^{k(\pi-\alpha+\xi)}dt \frac{t}{k(\pi+\alpha-\xi)}e^{-t^2},   & \alpha<\xi,  \\
        \frac{1}{\sqrt{\pi}}\int_{k(\pi-\alpha+\xi)}^{k(\pi+\alpha-\xi)}dt \frac{t}{k(\pi-\alpha+\xi)}e^{-t^2},   & \alpha>\xi, 
    \end{array}\right.\\
    &\leq\left\{\begin{array}{cc}
      \frac{1}{2\sqrt{\pi}}\frac{1}{k(\pi+\alpha-\xi)}e^{-k^2(\pi+\alpha-\xi)^2},   & \alpha<\xi,  \\
        \frac{1}{2\sqrt{\pi}} \frac{1}{k(\pi-\alpha+\xi)}e^{-k^2(\pi-\alpha+\xi)^2},   & \alpha>\xi, 
    \end{array}\right.\\
    &\leq\frac{1}{2\sqrt{\pi}k\alpha}e^{-k^2\alpha^2},
\end{split}
\end{equation}
which is solved by 
\begin{equation}
    k_2=\frac{1}{\sqrt{2}\alpha}\sqrt{W\left(\frac{1}{2\pi \epsilon^2}\right)}
\end{equation}

\begin{equation}
\begin{split}
        |R_3|&\leq \frac{1}{2\sqrt{\pi}k}\sum_{n=2}^{\infty}\frac{1}{n\pi+\alpha}e^{-k^2(n\pi+\alpha)^2}\\
        &\leq \frac{1}{2\sqrt{\pi}k}\int_1^{\infty}\frac{dt}{t\pi+\alpha}e^{-k^2(t\pi+\alpha)^2}\\
        &=\frac{1}{4k\pi^{3/2}}E_1\left[k^2(\pi+\alpha)^2\right]\\
        &\leq \frac{1}{4k\pi^{3/2}}\frac{1}{k(\pi+\alpha)}e^{-k^2(\pi+\alpha)^2}
\end{split}
\end{equation}

\begin{equation}
\begin{split}
        |R_4|&\leq \frac{1}{2\sqrt{\pi}k}\sum_{n=2}^{\infty}\frac{1}{n\pi-\alpha}e^{-k^2(n\pi-\alpha)^2}\\
        &\leq \frac{1}{2\sqrt{\pi}k}\int_1^{\infty}\frac{dt}{t\pi-\alpha}e^{-k^2(t\pi-\alpha)^2}\\
        &=\frac{1}{4k\pi^{3/2}}E_1\left[k^2(\pi-\alpha)^2\right]\\
        &\leq \frac{1}{4k\pi^{3/2}}\frac{1}{k(\pi-\alpha)}e^{-k^2(\pi-\alpha)^2}
\end{split}
\end{equation}

Suppose $k>1/\pi$, then 
\begin{equation}
\begin{split}
        |R_3|&\leq \frac{1}{4\sqrt{\pi}k\alpha}e^{-k^2\alpha^2},\\
        |R_4|&\leq \frac{1}{4\sqrt{\pi}k\alpha}e^{-k^2\alpha^2},
\end{split}
\end{equation}
so that 
\begin{equation}
\begin{split}
        \epsilon_\mathrm{tot}&\leq |R_1|+|R_2|+|R_3|+|R_4|\\
        &\leq \left(\frac{1}{2}+\sqrt{2}\right)\frac{1}{\sqrt{\pi}k\alpha}e^{-k^2\alpha^2},
\end{split}
\end{equation}
with solution
\begin{equation}
    k=\frac{1}{\sqrt{2}\alpha}\sqrt{W\left(\frac{9+4\sqrt{2}}{2\pi\epsilon^2}\right)}
\end{equation}
\end{proof} 

\section{Proof of Theorem \ref{thm:gpoly_even}}
\label{app:gpoly_even}
\begin{proof}
We begin by noting that for all $0\leq a<1$ and $x\in[-1,1]$,
\begin{equation}
    |\sqrt{1-a^2}U_{2n}(a) T_{2n}(x)|\leq 1.
\end{equation}
It follows that
\begin{equation}
\begin{split}
     \|\tilde{g}^+(k,a;x)-g^+_N(k,a;x)\|&\leq \frac{4}{\pi}\sum_{n=N/2+1}^{\infty}\frac{1}{2n}e^{-n^2/k^2}\\
      &\leq \frac{4}{\pi}\int_{N/2}^{\infty}dt ~\frac{1}{2t}e^{-t^2/k^2}\\
      &=\frac{1}{\pi}\mathrm{E_1}\left[\frac{N^2}{4k^2}\right]\\
      &\leq \frac{4}{\pi}\frac{k^2}{N^2}e^{-N^2/4k^2},
\end{split}
\end{equation}
where $E_1(x)\equiv\int_x^\infty dt~e^{-t}/t$ is the exponential integral function.
Thus for any $\epsilon>0$, choosing 
\begin{equation}
    N=2\ceil*{k\sqrt{W(1/(\pi\epsilon))}},
\end{equation}
guarantees that 
\begin{equation}
    \|\tilde{g}^+(k,a;x)-g^+_N(k,a;x)\|\leq \epsilon.
\end{equation}
\end{proof}

\clearpage

\section{Proof of Theorem \ref{thm:g_odd_proj}}
\label{app:g_odd_proj}
\begin{proof}
Letting $x=\cos\xi$, $a=\cos\alpha$, the Chebyshev projection of $g^-(k,a;x)$ can be expressed as
\begin{equation}
    \frac{2}{\pi}\int_0^\pi d\xi~\frac{1}{2}\left\{\mathrm{erf}[k(\pi-\alpha-\xi)]+\mathrm{erf}[k(\alpha-\xi)]\right\}\cos[(2n+1)\xi]=\tilde{a}_{2n+1}+\Delta a_{2n+1},
\end{equation}
where $\tilde{a}_{2n+1}$ is the Chebyshev coefficient of $\tilde{g}(k,a;x)$, given in Eq. \eqref{eq:an_odd}, and
\begin{equation}
\begin{split}
 \Delta a_{2n+1}\equiv \frac{2}{\pi}\frac{e^{-(2n+1)^2/4k^2}}{2n+1} &\Bigg\{\mathrm{Im}\left[e^{-i(2n+1)\alpha}\mathrm{erfc}\left(k(\pi-\alpha)+i\frac{2n+1}{2k}\right)\right]\\
 &+\mathrm{Im}\left[e^{i(2n+1)\alpha}\mathrm{erfc}\left(k\alpha+i\frac{2n+1}{2k}\right)\right]\Bigg\}.
 \end{split}
\end{equation}

We use the fact that $\mathrm{Im}(z)\leq |z|$ to bound
\begin{equation}
\begin{split}
       \mathrm{Im}\left\{e^{-i(2n+1)\alpha}\mathrm{erfc}\left[k(\pi-\alpha)+i\frac{2n+1}{2k}\right]\right\}&\leq  \left|\mathrm{erfc}\left[k(\pi-\alpha)+i\frac{2n+1}{2k}\right]\right|\\
        &\leq e^{(2n+1)^2/4k^2} \mathrm{erfc}[k(\pi-\alpha)]\\
        &\leq e^{(2n+1)^2/4k^2}\frac{1}{\sqrt{\pi}k(\pi-\alpha)}e^{-k^2(\pi-\alpha)^2}
\end{split}
\end{equation}

By the same token, 
\begin{equation}
    \mathrm{Im}\left\{e^{-i(2n+1)\alpha}\mathrm{erfc}\left[k\alpha+i\frac{2n+1}{2k}\right]\right\}\leq e^{(2n+1)^2/4k^2}\frac{1}{\sqrt{\pi}k\alpha}e^{-k^2\alpha^2}
\end{equation}
Therefore 
\begin{equation}
\begin{split}
        |\Delta a_{2n+1}|&\leq \frac{2}{\pi^{3/2}k(\pi-\alpha)}\frac{e^{-k^2(\pi-\alpha)^2}}{2n+1}+\frac{2}{\pi^{3/2}\alpha}\frac{e^{-k^2\alpha^2}}{2n+1}\\
        &\leq \frac{4}{\pi^{3/2}k\alpha}\frac{e^{-k^2\alpha^2}}{2n+1}.
\end{split}
\end{equation}
We see that $|\Delta a_{2n+1}|$ is maximized when $n=0$, so that choosing 
\begin{equation}
    k=\frac{1}{\sqrt{2}\alpha}\sqrt{W\left(\frac{32}{\pi^3\tilde{\epsilon}^2}\right)}
\end{equation}
guarantees that $|\Delta a_{2n+1}|\leq \tilde{\epsilon}$ for all $n\geq 0$.
\end{proof}

\section{Proof of Theorem \ref{thm:g_even_proj}}
\label{app:g_even_proj}
\begin{proof}
Adopting $x=\cos\xi$, $a=\cos\alpha$, we find that the Chebyshev projection of $g^+(k,a;x)$ can be expressed as
\begin{equation}
    \frac{\delta_n}{\pi}\int_0^\pi d\xi~\frac{1}{2}\left\{\mathrm{erf}[k(\pi-\alpha-\xi)]+\mathrm{erf}[k(\alpha-\xi)]\right\}\cos[(2n+1)\xi]=\tilde{a}_{2n}+\Delta a_{2n}.
\end{equation}
When $n=0$, we have
\begin{equation}
    \Delta a_0=-\frac{\alpha}{\pi}\left\{\mathrm{erfc}(k\alpha)+\mathrm{erfc}[k(\pi-\alpha)]\right\}+\frac{1}{k\pi^{3/2}}\left[e^{-k^2(\pi-\alpha)^2}+e^{-k^2\alpha^2}\right]
\end{equation}
where 
\begin{equation}
\begin{split}
\mathrm{erfc}(k\alpha)+\mathrm{erfc}[k(\pi-\alpha)]&\leq 2\mathrm{erfc}(k\alpha)\\
&\leq \frac{2}{\sqrt{\pi}k\alpha} e^{-k^2\alpha^2},
\end{split}
\end{equation}
so that
\begin{equation}
    |\Delta a_0|\leq \frac{4}{\pi^{3/2}k}e^{-k^2\alpha^2}\leq \frac{2}{\sqrt{\pi}k\alpha}e^{-k^2\alpha^2}.
\end{equation}
For $n\geq 1$, we find
\begin{equation}
    \begin{split}
        \Delta a_{2n}&=\frac{2}{\pi}\frac{e^{-n^2/k^2}}{2n}\Bigg\{\mathrm{Im}\left[e^{-2in\alpha}\mathrm{erf}\left(k(\pi-\alpha)+i\frac{n}{k}\right)\right]+\mathrm{Im}\left[e^{2in\alpha}\mathrm{erf}\left(k\alpha+i\frac{n}{k}\right)\right]\Bigg\}
    \end{split}
\end{equation}

We use the fact that $\mathrm{Im}(z)\leq |z|$ to bound
\begin{equation}
\begin{split}
       \mathrm{Im}\left\{e^{-2in\alpha}\mathrm{erfc}\left[k(\pi-\alpha)+i\frac{n}{k}\right]\right\}&\leq  \left|\mathrm{erfc}\left[k(\pi-\alpha)+i\frac{n}{k}\right]\right|\\
        &\leq e^{n^2/k^2} \mathrm{erfc}[k(\pi-\alpha)]\\
        &\leq e^{n^2/k^2}\frac{1}{\sqrt{\pi}k(\pi-\alpha)}e^{-k^2(\pi-\alpha)^2}
\end{split}
\end{equation}

By the same token, 
\begin{equation}
    \mathrm{Im}\left\{e^{-2in\alpha}\mathrm{erfc}\left[k\alpha+i\frac{n}{k}\right]\right\}\leq e^{n^2/k^2}\frac{1}{\sqrt{\pi}k\alpha}e^{-k^2\alpha^2}
\end{equation}
So that 
\begin{equation}
    |\Delta a_{2n}| \leq \frac{4}{\pi^{3/2}k\alpha}\frac{e^{-k^2\alpha^2}}{2n}, ~~~n\geq 1.
\end{equation}
Combining this result with that obtained for $n=0$, it follows that for all $\tilde{\epsilon}>0$, choosing 
\begin{equation}
    k=\frac{1}{\sqrt{2}\alpha}\sqrt{W\left(\frac{8}{\pi\tilde{\epsilon}^2}\right)},
\end{equation}
guarantees that $|\Delta a_{2n}|\leq \tilde{\epsilon}$ for all $n\geq 0$.
\end{proof}

\clearpage
\end{widetext}

\end{document}